\documentclass[12pt]{article}

\usepackage[english]{babel}
\usepackage{amsmath,amssymb,amsthm,mathtools,mathrsfs,bm}
\usepackage{graphicx,xcolor,booktabs,array,enumitem,cite,microtype}
\usepackage{tikz}
\usetikzlibrary{decorations.markings,decorations.pathmorphing,patterns,
  positioning,backgrounds,arrows.meta,calc}
\usepackage[colorlinks=true,linkcolor=blue,citecolor=blue,urlcolor=blue,
  unicode=true]{hyperref}

\tikzset{partial ellipse/.style args={#1:#2:#3}{insert path={+ (#1:#3) arc (#1:#2:#3)}}}
\tikzset{->-/.style={decoration={markings,mark=at position #1 with {\arrow{>}}},postaction={decorate}}}
\tikzset{-<-/.style={decoration={markings,mark=at position #1 with {\arrow{<}}},postaction={decorate}}}

\hypersetup{
  pdftitle={Long-time asymptotics of a full arbitrary-genus dark soliton gas for the defocusing nonlinear Schrodinger equation},
  pdfauthor={Anonymous Author(s)},
  pdfsubject={Defocusing NLS full soliton gas and Riemann-Hilbert asymptotics},
  pdfkeywords={soliton gas, defocusing NLS, Riemann-Hilbert problem, theta function}
}
\numberwithin{equation}{section}

\newtheorem{theorem}{Theorem}[section]
\newtheorem{lemma}[theorem]{Lemma}

\newtheorem{proposition}[theorem]{Proposition}
\newtheorem{definition}[theorem]{Definition}
\newtheorem{assumption}[theorem]{Assumption}
\newtheorem{RHP}[theorem]{Riemann--Hilbert problem}
\theoremstyle{remark}
\newtheorem{remark}[theorem]{Remark}

\newcommand{\C}{\mathbb C}
\newcommand{\R}{\mathbb R}

\newcommand{\ii}{\mathrm i}
\newcommand{\dd}{\mathrm d}
\newcommand{\ee}{\mathrm e}
\newcommand{\Rea}{\operatorname{Re}}
\newcommand{\Ima}{\operatorname{Im}}
\newcommand{\diag}{\operatorname{diag}}
\newcommand{\Res}{\operatorname*{Res}}

\newcommand{\abs}[1]{\lvert#1\rvert}
\newcommand{\ol}{\overline}

\newcommand{\ThetaN}{\Theta^{[N]}}
\newcommand{\RN}{R_N}
\newcommand{\PN}{P_N}
\newcommand{\QN}{Q_N}

\newcommand{\CN}{\mathcal C_N}
\newcommand{\JN}{J_N}
\newcommand{\tauN}{\tau^{[N]}}
\newcommand{\AN}{A_N}
\newcommand{\TN}{\mathcal T_N}

\title{\Large\bfseries Long-time asymptotics of a full arbitrary-genus dark soliton gas\\[1mm]
	for the defocusing nonlinear Schr\"odinger equation}
\author{\hspace{0.6 cm}{}Dedi Yan$^{b}$, Xianguo Geng$^{a,b}$\footnote{\footnotesize
		Corresponding author.{\sl E-mail address}: xggeng@zzu.edu.cn}, Mingming Chen$^{b}$\\ 
	\leftline{\hspace{0.6 cm}{\small{\sl $^{a}$ School of Mathematics and Statistics, North China University of Water Resources }}}\\
	\leftline{\hspace{0.6 cm}{\small{\sl \quad and Electric Power, Zhengzhou, Henan 450011, People's Republic of China}}}\\
	\leftline{\hspace{0.6 cm}{\small{\sl $^{b}$ School of Mathematics and Statistics, Zhengzhou University, 100 Kexue Road, Zhengzhou, }}}\\
	\leftline{\hspace{0.6 cm}{\small{\sl \quad Henan 450001, People's Republic of China}}}}
\date{}

\begin{document}
\maketitle

\begin{abstract}
We introduce a full arbitrary-genus dark soliton gas for the defocusing nonlinear Schr\"odinger equation with finite-density boundary conditions. Starting from a generalized meromorphic Riemann--Hilbert problem with two alternating residue families on each unit-circle arc, we derive an exact thermodynamic limit whose jump matrix contains two nonzero continuum densities. The limiting Riemann--Hilbert problem is uniquely solvable. In contrast with the half dark-soliton gas, every spectral arc of the full gas carries both oscillatory exponentials. We analyze the resulting problem by the Deift--Zhou nonlinear steepest-descent method on a fixed genus-$N$ spectral curve. The moving point in each mixed sector is a stationary factorization-switching point, not a branch point. The active arc is split into two parts and opened crosswise, while lenses are opened around every remaining arc. After removal of exponentially small lens jumps, the model contour therefore retains all $N$ spectral arcs in every self-similar sector. A quotient-curve zero-counting argument proves strict monotonicity of the characteristic velocity and the global ordering of all endpoint velocities, so the self-similar line is divided into $2N+1$ nonempty sectors. The leading term is an $N$-dimensional Riemann-theta finite-gap solution. The error is $O(t^{-1})$ in the $N+1$ pure sectors and $O(t^{-1/2})$ in the $N$ mixed sectors, uniformly away from the critical rays.
\end{abstract}


\tableofcontents

\section{Introduction}\label{sec:intro}

We consider the defocusing nonlinear Schr\"odinger equation in finite-density normalization,
\begin{equation}\label{eq:dnls}
 \ii q_t+q_{xx}-2\bigl(\abs{q}^2-1\bigr)q=0,
 \qquad \abs{q(x,t)}\longrightarrow1\quad (\abs{x}\to\infty).
\end{equation}
Its localized coherent structures are dark solitons, namely density depressions propagating on a nonzero background.  Their basic form goes back to Tsuzuki \cite{Tsuzuki1971}; physical realizations and surveys in Bose--Einstein condensates and nonlinear optics may be found in \cite{BurgerEtAl1999,Frantzeskakis2010,KevrekidisEtAl2015,PitaevskiiStringari2003,KivsharAgrawal2003}.  The inverse-scattering and long-time stability theory for dark multi-solitons under finite-density boundary conditions is well developed; see \cite{FaddeevTakhtajan1987,ZakharovShabat1973,CuccagnaJenkins2016}.  Related dNLS hydrodynamic Riemann problems and optical experiments are studied in \cite{ElEtAl1995,Jenkins2015,XuEtAl2017,SprengerHoeferEl2018,BendahmaneEtAl2022,WangYan2025}.  The associated spectral theory is naturally written in the Joukowski variable
\begin{equation}\label{eq:j-map}
 k=k(z)=\frac12\left(z+z^{-1}\right),
\end{equation}
which separates the two sheets of the nonzero-background spectral plane.  The price of this uniformization is the strong involution $z\mapsto z^{-1}$ and a prescribed singularity of the matrix Riemann--Hilbert unknown at $z=0$; these features also play a central role in the arbitrary-genus half dark-soliton gas \cite{BertolaWangYanZhu2026}.

The notion of a soliton gas goes back to Zakharov's kinetic description \cite{Zakharov1971}.  Its thermodynamic and kinetic theory was developed in \cite{El2003,ElKamchatnov2005,ElKamchatnovPavlovZykov2011,ElTovbis2020,CongyElRoberti2021}; for dark dNLS gases and condensates, see \cite{TovbisWang2025}.  On the Riemann--Hilbert side, primitive potentials provide deterministic full-jump realizations of soliton-gas-type states \cite{DyachenkoZakharovZakharov2016,Nabelek2020}.  Rigorous nonlinear steepest-descent asymptotics for soliton gases were initiated for KdV in \cite{GirottiGravaJenkinsMcLaughlin2021} and further developed in \cite{GirottiEtAl2023,GengYanJia2025,WangZhuZhu2026,HanZhangDong2025,ZhangLing2025}.  Condensate limits and full-gas scattering backgrounds for focusing NLS have recently been investigated in \cite{GkogkouMazzucaMcLaughlin2025,GravaJenkinsZhangZhang2026a,GravaJenkinsZhangZhang2026b}.  Arbitrary-genus and nonzero-background gas asymptotics also appear in \cite{YanGengWang2026,YanGengJiao2026}.  A two-density continuum limit obtained from two interlacing residue families was developed for a new KdV soliton gas in \cite{YanGengWei2026}, and the corresponding full-jump mechanism was extended to the Camassa--Holm setting in \cite{FullCH2026}.

A half dark-soliton gas is obtained by filling finitely many arcs of the unit circle with one family of dark-soliton poles \cite{BertolaWangYanZhu2026}.  Its continuum Riemann--Hilbert problem has one triangular jump on each filled arc.  The corresponding long-time analysis produces a genus cascade because only part of the spectrum contributes to the leading model in a given self-similar sector.  The full gas studied here has a different continuum limit: two interlacing pole families produce two nonzero spectral densities on every arc, in direct analogy with the full-jump constructions in \cite{YanGengWei2026,FullCH2026}.  Thus the upper jump is
\begin{equation}\label{eq:intro-full-jump}
 J_M(z)=\frac{1}{1+r(z)\rho(z)}
 \begin{pmatrix}
 1-r(z)\rho(z)&-2\ii\rho(z)\ee^{-\ii\Phi(z;x,t)}\\
 -2\ii r(z)\ee^{\ii\Phi(z;x,t)}&1-r(z)\rho(z)
 \end{pmatrix},
\end{equation}
where
\begin{equation}\label{eq:phase-intro}
 \Phi(z;x,t)=x(z-z^{-1})-t(z^2-z^{-2}).
\end{equation}
Both oscillatory exponentials occur simultaneously.  Therefore one must use two different triangular factorizations, selected according to the sign of the fixed-genus phase.  The long-time analysis is carried out by the Deift--Zhou nonlinear steepest-descent method \cite{DeiftZhou1993,DeiftItsZhou1997,DeiftZhou2002}, with a fixed-curve $g$-function and a region-dependent Szeg\H{o} factor.  We use standard Fredholm and vanishing-lemma theory for Riemann--Hilbert problems \cite{Zhou1989}, and the finite-gap model is written with the classical theta-function machinery of \cite{FarkasKra1992,Fay1973}.  Closely related nonlinear steepest-descent analyses for NLS equations with nonzero or step-like backgrounds include \cite{BiondiniMantzavinos2017,BiondiniLiMantzavinos2021,BoutetItsKotlyarov2009,BoutetLenellsShepelsky2021,BoutetLenellsShepelsky2022}.  Hyperelliptic and step-like steepest-descent constructions for KdV-, mKdV-, and Gross--Pitaevskii-type systems that are useful points of comparison are developed in \cite{GrunertTeschl2009,GravaMinakov2020,KotlyarovMinakov2012,GengWangChen2021}.

Let
\begin{equation}\label{eq:endpoints-intro}
 \eta_m=\ee^{\ii\varphi_m},\qquad
 0<\varphi_1<\varphi_2<\cdots<\varphi_{2N}<\pi,
\end{equation}
and define the $N$ upper spectral arcs
\begin{equation}\label{eq:bands-intro}
 \Gamma_j=(\eta_{2j-1},\eta_{2j}),\qquad j=1,\ldots,N,
\end{equation}
with their reciprocal reflections $\ol\Gamma_j=(\eta_{2j}^{-1},\eta_{2j-1}^{-1})$ in the lower half-plane.  Put
\begin{equation}\label{eq:e-intro}
 e_m=\Rea\eta_m,
 \qquad 1>e_1>e_2>\cdots>e_{2N}>-1.
\end{equation}
The fixed quotient curve in the $k$-plane is
\begin{equation}\label{eq:curve-intro}
 \CN:\qquad y^2=(k^2-1)\prod_{m=1}^{2N}(k-e_m),
\end{equation}
which has genus $N$.

The normalized second-kind differentials introduced in Section~\ref{sec:g} determine the velocity function
\begin{equation}\label{eq:V-intro}
 V_N(z)=\frac{2\QN(z)}{\PN(z)}.
\end{equation}
We set
\begin{equation}\label{eq:v-intro}
 v_m=V_N(\eta_m),\qquad m=1,\ldots,2N,
\end{equation}
Theorem~\ref{thm:velocity-ordering} below proves, for every strictly ordered endpoint set \eqref{eq:e-intro}, that
\begin{equation}\label{eq:v-order-intro}
 v_1>v_2>\cdots>v_{2N}.
\end{equation}
Thus no additional velocity-ordering assumption is imposed.
The self-similar line $\xi=x/t$ is then partitioned into $2N+1$ open sectors.  The pure sectors are
\begin{align}
 \mathcal P_0&=(v_1,+\infty),\label{eq:P0-intro}\\
 \mathcal P_\ell&=(v_{2\ell+1},v_{2\ell}),\qquad \ell=1,\ldots,N-1,\label{eq:Pl-intro}\\
 \mathcal P_N&=(-\infty,v_{2N}),\label{eq:PN-intro}
\end{align}
and the mixed sectors are
\begin{equation}\label{eq:Ml-intro}
 \mathcal M_\ell=(v_{2\ell},v_{2\ell-1}),\qquad \ell=1,\ldots,N.
\end{equation}
In $\mathcal P_\ell$, the first $\ell$ arcs use the $r$-factorization and the remaining $N-\ell$ arcs use the $\rho$-factorization.  In $\mathcal M_\ell$, the arcs with $j<\ell$ use the $r$-factorization, the arcs with $j>\ell$ use the $\rho$-factorization, and $\Gamma_\ell$ is split at the unique point $\alpha_\ell(\xi)$ satisfying
\begin{equation}\label{eq:alpha-intro}
 V_N(\alpha_\ell(\xi))=\xi,
 \qquad \alpha_\ell(\xi)\in\Gamma_\ell.
\end{equation}

The main geometric point of this paper is that lenses are opened around \emph{all} the arcs, not only around the arcs preceding the active one.  In a mixed sector, the two lens systems on $\Gamma_\ell$ are joined crosswise at $\alpha_\ell$, while every $\Gamma_j$, $j>\ell$, retains its own $\rho$-lens.  After the lens-lip jumps become exponentially small, the central jump $-\ii\sigma_1$ survives on the whole of every upper arc $\Gamma_j$.  Hence the outer model has fixed genus $N$ in every sector.

To state the leading term, let $\ThetaN(\cdot;\tauN)$ be the Riemann theta function of \eqref{eq:curve-intro}, let $\JN$ be the normalized Abel map, and set
\begin{equation}\label{eq:AN-intro}
 \AN=\frac12\sum_{j=0}^{N}\bigl(e_{2j}-e_{2j+1}\bigr),
 \qquad e_0=1,\quad e_{2N+1}=-1.
\end{equation}
Then, away from the critical rays, the leading term has the universal form
\begin{equation}\label{eq:main-intro}
 q(x,t)=\AN\ee^{-2\ii t g_{N,\infty}(\xi)}\delta_{N,\infty}(\xi)^{-2}
 \TN\!\left(\frac{t\Omega_N(\xi)+\Delta_N(\xi)}{2\pi}\right)
 +\mathcal E_N(x,t),
\end{equation}
where
\begin{equation}\label{eq:TN-intro}
 \TN(\mathbf Z)=
 \frac{\ThetaN(-\mathbf Z+2\JN(\infty_+);\tauN)\ThetaN(0;\tauN)}
 {\ThetaN(2\JN(\infty_+);\tauN)\ThetaN(\mathbf Z;\tauN)}.
\end{equation}
Here $\Omega_N(\xi)$ is the vector of real $b$-periods of the fixed-curve phase differential, defined componentwise in \eqref{eq:OmegaNm-explicit}--\eqref{eq:OmegaN-vector}.  The phase vector $\Delta_N$ and the scalar $\delta_{N,\infty}$ are determined by the region-dependent Szeg\H{o} function.  The error is $O(t^{-1})$ in $\mathcal P_\ell$ and $O(t^{-1/2})$ in $\mathcal M_\ell$.

The paper is organized as follows.  Section~\ref{sec:RHP} formulates the full arbitrary-genus gas Riemann--Hilbert problem.  Section~\ref{sec:g} constructs the fixed genus-$N$ $g$-function and proves strict monotonicity and the complete ordering of its endpoint velocities.  Section~\ref{sec:factor} gives the two factorizations and the Szeg\H{o} function.  Section~\ref{sec:lens} performs the all-band lens opening, including the crossed opening on the active arc.  Section~\ref{sec:model} solves the fixed genus-$N$ model problem.  The long-time theorem and its proof are given in Sections~\ref{sec:theorem} and~\ref{sec:proof}.

\section{From two interlacing pole families to the full arbitrary-genus dNLS gas}\label{sec:RHP}

We reserve $N$ for the number of spectral arcs, and hence for the genus of the fixed quotient curve.  The thermodynamic limit is taken with $N$ fixed and with the number of poles on each arc tending to infinity.  More precisely, if $n_j$ denotes the number of pole pairs of each type on $\Gamma_j$, then
\begin{equation}\label{eq:total-discrete-number}
 \min_{1\le j\le N}n_j\longrightarrow\infty,
 \qquad
 \mathcal N:=\sum_{j=1}^{N}n_j\longrightarrow\infty.
\end{equation}
Thus the phrase ``$\mathcal N\to\infty$'' below does not change the prescribed genus $N$.

Throughout the paper every upper arc $\Gamma_j$ is oriented counterclockwise, from $\eta_{2j-1}$ to $\eta_{2j}$.  Its reciprocal image $\ol\Gamma_j$ is oriented clockwise, from $\eta_{2j}^{-1}$ to $\eta_{2j-1}^{-1}$.  The $+$ boundary value is always taken from the left side of the oriented contour.  These conventions are used in all residue-to-jump calculations and in the lower-arc symmetry formulas.

We use
\begin{equation*}
 \sigma_1=\begin{pmatrix}0&1\\1&0\end{pmatrix},\qquad
 \sigma_2=\begin{pmatrix}0&-\ii\\ \ii&0\end{pmatrix},\qquad
 \sigma_3=\begin{pmatrix}1&0\\0&-1\end{pmatrix},
\end{equation*}
and the elementary triangular matrices
\begin{equation}\label{eq:LU-def}
 L(a):=I+aE_{21}=\begin{pmatrix}1&0\\a&1\end{pmatrix},
 \qquad
 U(b):=I+bE_{12}=\begin{pmatrix}1&b\\0&1\end{pmatrix}.
\end{equation}

\subsection{Alternating discrete spectrum on all \texorpdfstring{$N$}{N} arcs}

Write
\begin{equation}\label{eq:angle-endpoints}
 \hat{a}_j:=\arg\eta_{2j-1},\qquad \hat{b}_j:=\arg\eta_{2j},\qquad
 h_j:=\frac{\hat{b}_j-\hat{a}_j}{n_j},\qquad j=1,\ldots,N.
\end{equation}
On the $j$th upper arc we choose two interlacing sets
\begin{align}
 w_{m,j}&:=\exp\!\left(\ii\left[\hat{a}_j+\left(m-\frac34\right)h_j\right]\right),
 &&m=1,\ldots,n_j,\label{eq:w-points}\\
 z_{m,j}&:=\exp\!\left(\ii\left[\hat{a}_j+\left(m-\frac14\right)h_j\right]\right),
 &&m=1,\ldots,n_j.\label{eq:z-points}
\end{align}
The blue points $w_{m,j}$ carry upper-triangular, or $\rho$-type, residues, whereas the red points $z_{m,j}$ carry lower-triangular, or $r$-type, residues.  In the counterclockwise order on $\Gamma_j$,
\begin{equation}\label{eq:interlacing-order}
 \eta_{2j-1}\prec w_{1,j}\prec z_{1,j}\prec w_{2,j}\prec z_{2,j}
 \prec\cdots\prec w_{n_j,j}\prec z_{n_j,j}\prec\eta_{2j}.
\end{equation}
The lower-half-plane points are the reciprocal-conjugate points $w_{m,j}^{-1}=\ol{w_{m,j}}$ and $z_{m,j}^{-1}=\ol{z_{m,j}}$.  The complete alternating configuration on representative arcs is shown in Figure~\ref{fig:alternating-discrete-spectrum}.

\begin{figure}[htbp]
\centering
\begin{tikzpicture}[scale=2.65,line cap=round,line join=round,>=Stealth]
  \draw[gray!65,dashed,line width=.55pt] (0,0) circle (1);
  \fill (0,0) circle (.65pt);
  \node[scale=.54,below] at (0,-.06) {$0$};
  \node[scale=.56,right] at (1,0) {$1$};
  \node[scale=.56,left] at (-1,0) {$-1$};
  \tikzset{specarc/.style={black,line width=1.05pt},
           rdot/.style={red!78!black},
           rhodot/.style={blue!76!black}}
  % Four representative arcs: Gamma_1, Gamma_2, Gamma_{N-1}, Gamma_N.
  \foreach \aa/\bb/\lab in {11/29/1,45/64/2,91/111/{N-1},132/153/N}{
    \draw[specarc] (\aa:1) arc[start angle=\aa,end angle=\bb,radius=1];
    \draw[specarc] (-\bb:1) arc[start angle=-\bb,end angle=-\aa,radius=1];
    \foreach \m in {1,...,8}{
      \pgfmathsetmacro{\ang}{\aa+(\m-.5)*(\bb-\aa)/8}
      \ifodd\m
        \fill[rhodot] (\ang:1) circle (.88pt);
        \fill[rhodot] (-\ang:1) circle (.88pt);
      \else
        \fill[rdot] (\ang:1) circle (.88pt);
        \fill[rdot] (-\ang:1) circle (.88pt);
      \fi
    }
    \node[scale=.50] at ({(\aa+\bb)/2}:1.18) {$\Gamma_{\lab}$};
  }
  \node[scale=.68] at (75:1.04) {$\cdots$};
  \node[scale=.68] at (-75:1.04) {$\cdots$};
  \node[scale=.48,align=left,anchor=west] at (1.18,.43)
  {\textcolor{red!78!black}{$\bullet$} $z_{m,j}$: $r$-type residue\\
   \textcolor{blue!76!black}{$\bullet$} $w_{m,j}$: $\rho$-type residue};
  \draw[->,gray!75,line width=.45pt] (8:.88) arc[start angle=8,end angle=31,radius=.88];
  \node[scale=.46] at (20:.68) {counterclockwise};
\end{tikzpicture}
\caption{Two alternating discrete pole families on the $N$ upper unit-circle arcs and their reciprocal-conjugate copies on the lower arcs.  The drawing shows representative arcs $\Gamma_1,\Gamma_2,\Gamma_{N-1},\Gamma_N$; the same red-blue alternation is imposed on every $\Gamma_j$.  Red points carry lower-triangular $r$-type residues and blue points carry upper-triangular $\rho$-type residues.}
\label{fig:alternating-discrete-spectrum}
\end{figure}
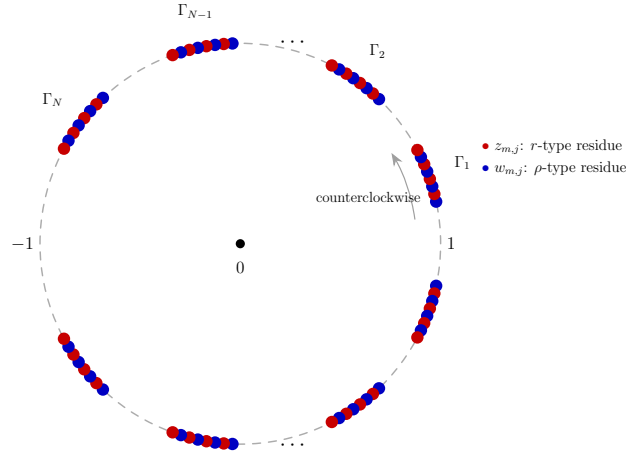

\subsection{A generalized finite-pole dNLS Riemann--Hilbert problem}

The finite precursor below is a primitive-potential-type meromorphic problem.  It is not the standard reflectionless dark $\mathcal N$-soliton RHP, because the latter has only one triangular residue family on the upper unit circle.

\begin{RHP}[Two-family finite-pole precursor]\label{rhp:finite-full}
Find a $2\times2$ matrix $M^{(\mathbf n)}(z)=M^{(\mathbf n)}(z;x,t)$, where $\mathbf n=(n_1,\ldots,n_N)$, with the following properties.
\begin{enumerate}[label=(\roman*)]
\item $M^{(\mathbf n)}$ is meromorphic in $\C$ with simple poles at
\begin{equation}\label{eq:finite-pole-set}
 \bigcup_{j=1}^{N}\bigcup_{m=1}^{n_j}
 \{z_{m,j},w_{m,j},z_{m,j}^{-1},w_{m,j}^{-1}\}.
\end{equation}
\item As $z\to\infty$ and $z\to0$,
\begin{equation}\label{eq:finite-normalization}
 M^{(\mathbf n)}(z)=I+O(z^{-1}),\qquad
 M^{(\mathbf n)}(z)=\frac{\sigma_1}{z}+O(1),
\end{equation}
respectively.
\item The dNLS reductions are
\begin{equation}\label{eq:finite-symmetries}
 M^{(\mathbf n)}(z)=\sigma_1\ol{M^{(\mathbf n)}(\ol z)}\sigma_1
 =z^{-1}M^{(\mathbf n)}(z^{-1})\sigma_1.
\end{equation}
\item At the upper red and blue poles,
\begin{align}
 \Res_{z=z_{m,j}}M^{(\mathbf n)}(z)
 &=\lim_{z\to z_{m,j}}M^{(\mathbf n)}(z)
 \begin{pmatrix}0&0\\
 \chi_{m,j}\ee^{\ii\Phi(z_{m,j};x,t)}&0
 \end{pmatrix},\label{eq:red-residue}\\
 \Res_{z=w_{m,j}}M^{(\mathbf n)}(z)
 &=\lim_{z\to w_{m,j}}M^{(\mathbf n)}(z)
 \begin{pmatrix}0&
 \mu_{m,j}\ee^{-\ii\Phi(w_{m,j};x,t)}\\0&0
 \end{pmatrix}.
 \label{eq:blue-residue}
\end{align}
At the reciprocal lower poles,
\begin{align}
 \Res_{z=z_{m,j}^{-1}}M^{(\mathbf n)}(z)
 &=\lim_{z\to z_{m,j}^{-1}}M^{(\mathbf n)}(z)
 \begin{pmatrix}0&
 \ol{\chi_{m,j}}\ee^{-\ii\Phi(z_{m,j}^{-1};x,t)}\\0&0
 \end{pmatrix},\label{eq:red-lower-residue}\\
 \Res_{z=w_{m,j}^{-1}}M^{(\mathbf n)}(z)
 &=\lim_{z\to w_{m,j}^{-1}}M^{(\mathbf n)}(z)
 \begin{pmatrix}0&0\\
 \ol{\mu_{m,j}}\ee^{\ii\Phi(w_{m,j}^{-1};x,t)}&0
 \end{pmatrix}.
 \label{eq:blue-lower-residue}
\end{align}
\end{enumerate}
\end{RHP}

The two reductions in \eqref{eq:finite-symmetries} are compatible with the residue conditions if and only if
\begin{equation}\label{eq:discrete-reduction-coefficients}
 \ol{\chi_{m,j}}=-z_{m,j}^{-2}\chi_{m,j},\qquad
 \ol{\mu_{m,j}}=-w_{m,j}^{-2}\mu_{m,j}.
\end{equation}
Equivalently,
\begin{equation}\label{eq:real-discrete-weights}
 \chi_{m,j}=\ii z_{m,j}\kappa_{m,j},\qquad
 \mu_{m,j}=\ii w_{m,j}\lambda_{m,j},
 \qquad \kappa_{m,j},\lambda_{m,j}\in\R.
\end{equation}
No sign restriction on $\kappa_{m,j}$ or $\lambda_{m,j}$ is needed at the level of the meromorphic RHP.

Whenever RHP~\ref{rhp:finite-full} is solvable, the usual Lax-pair/Liouville argument gives a dNLS potential through
\begin{equation}\label{eq:finite-reconstruction}
 q^{(\mathbf n)}(x,t)=\lim_{z\to\infty}zM^{(\mathbf n)}_{21}(z;x,t).
\end{equation}

\subsection{Pole removal and the interlacing product}

For every $j$, define
\begin{equation}\label{eq:Bnj}
 B_{\mathbf n,j}(z):=\prod_{m=1}^{n_j}
 \frac{z-w_{m,j}}{z-z_{m,j}}.
\end{equation}
We assume that 
\begin{align}
 \chi_{m,j}
 &=a_{j}(z_{m,j})
 \Res_{z=z_{m,j}}B_{\mathbf n,j}(z),\label{eq:chi-interpolation}\\
 \mu_{m,j}
 &=-\frac{b_{j}(w_{m,j})}
 {1-a_{j}(w_{m,j})b_{j}(w_{m,j})}
 \Res_{z=w_{m,j}}B_{\mathbf n,j}(z)^{-1}.\label{eq:mu-interpolation}
\end{align}
The denominator in \eqref{eq:mu-interpolation} is essential: it is the finite-$\mathbf n$ cross term produced by the ordered product of the two triangular pole-removing factors.  We further assume
\[
a_j(z)=O(z-\eta_{2j-1}),\qquad k\to \eta_{2j-1},\ a_j(z)=O(z-\eta_{2j}),\qquad k\to \eta_{2j},
\]
\[
b_j(z)=O(z-\eta_{2j-1}),\qquad k\to \eta_{2j-1},\ b_j(z)=O(z-\eta_{2j}),\qquad k\to \eta_{2j},.
\]  Conditions \eqref{eq:discrete-reduction-coefficients} are imposed on the discrete values in \eqref{eq:chi-interpolation}--\eqref{eq:mu-interpolation}; in the limit they become the dNLS reduction conditions on the continuum densities.

Let $\mathscr C_j$ be a counterclockwise loop enclosing $\Gamma_j$ and no other spectral arc, and let $\mathscr C_j^\sharp$ enclose $\ol\Gamma_j$.  Set
\begin{equation}\label{eq:Ffinite-upper}
 F_{\mathbf n,j}(z):=
 L\!\left(-a_{j}(z)B_{\mathbf n,j}(z)\ee^{\ii\Phi(z;x,t)}\right)
 U\!\left(b_{j}(z)B_{\mathbf n,j}(z)^{-1}\ee^{-\ii\Phi(z;x,t)}\right).
\end{equation}
For a function $f$ defined near the upper arcs, use the Schwarz-reflected notation
\begin{equation}\label{eq:sharp-def-prelimit}
 f^\sharp(z):=\ol{f(\ol z)},\qquad z\text{ near a lower arc},
\end{equation}
and define
\begin{equation}\label{eq:Ffinite-lower}
 F_{\mathbf n,j}^\sharp(z)
 :=\sigma_1\ol{F_{\mathbf n,j}(\ol z)}\sigma_1
 =U\!\left(-a_{j}^\sharp(z)B_{\mathbf n,j}^\sharp(z)
 \ee^{-\ii\Phi(z;x,t)}\right)
 L\!\left(b_{j}^\sharp(z)(B_{\mathbf n,j}^\sharp(z))^{-1}
 \ee^{\ii\Phi(z;x,t)}\right),
\end{equation}
where $B_{\mathbf n,j}^\sharp(z)=\ol{B_{\mathbf n,j}(\ol z)}$.  This is exactly the factor that cancels the two lower residue families in \eqref{eq:red-lower-residue}--\eqref{eq:blue-lower-residue}.
The pole-removing transformation is
\begin{equation}\label{eq:Dfinite}
 D^{(\mathbf n)}(z):=M^{(\mathbf n)}(z)
 \begin{cases}
 F_{\mathbf n,j}(z),&z\text{ inside }\mathscr C_j,\\
 F_{\mathbf n,j}^\sharp(z),&z\text{ inside }\mathscr C_j^\sharp,\\
 I,&\text{elsewhere}.
 \end{cases}
\end{equation}
To check the cancellation, write $\varepsilon=z-w_{m,j}$ near a blue pole.  Since $B_{\mathbf n,j}(z)=B'_{\mathbf n,j}(w_{m,j})\varepsilon+O(\varepsilon^2)$, the singular coefficient of $F_{\mathbf n,j}$ is
\[
 b_{j}(w_{m,j})\Res_{z=w_{m,j}}B_{\mathbf n,j}(z)^{-1}
 \ee^{-\ii\Phi(w_{m,j};x,t)}E_{12}.
\]
The constant lower-right entry of the ordered product in \eqref{eq:Ffinite-upper} is
$1-a_{j}(w_{m,j})b_{j}(w_{m,j})$.  Hence, if
$M^{(\mathbf n)}=M_0(I+\mu_{m,j}\ee^{-\ii\Phi(w_{m,j})}E_{12}/\varepsilon)$ locally, the coefficient of $\varepsilon^{-1}$ in $M^{(\mathbf n)}F_{\mathbf n,j}$ is
\[
 \left[b_{j}\Res B_{\mathbf n,j}^{-1}
 +\mu_{m,j}(1-a_{j}b_{j})\right]
 \ee^{-\ii\Phi(w_{m,j})}M_0E_{12}=0
\]
by \eqref{eq:mu-interpolation}.  At a red pole, $B_{\mathbf n,j}$ has a simple pole and $B_{\mathbf n,j}^{-1}$ a simple zero; the singular coefficient is
$-a_{j}(z_{m,j})\Res B_{\mathbf n,j}\,\ee^{\ii\Phi(z_{m,j})}E_{21}$, which cancels \eqref{eq:red-residue} by \eqref{eq:chi-interpolation}.  The reflected calculation treats the two lower families.  Thus all poles are removed exactly.  The only jumps of $D^{(\mathbf n)}$ are on the loops $\mathscr C_j\cup\mathscr C_j^\sharp$, with upper-loop jump $F_{\mathbf n,j}$ and lower-loop jump $F_{\mathbf n,j}^\sharp$.

Denote by $\Gamma_j^{\rm cl}$ the closed unit-circle arc from $\eta_{2j-1}$ to $\eta_{2j}$.

\begin{proposition}[Limit of the alternating product]\label{prop:B-limit}
For every $j=1,\ldots,N$,
\begin{equation}\label{eq:B-limit}
 B_{\mathbf n,j}(z)\longrightarrow
 \beta_j(z):=
 \left(\frac{z-\eta_{2j-1}}{z-\eta_{2j}}\right)^{1/2},
 \qquad z\in\C\setminus\Gamma_j^{\rm cl},
\end{equation}
uniformly on compact subsets.  The branch is fixed by $\beta_j(z)\to1$ as $z\to\infty$, and
\begin{equation}\label{eq:beta-boundary}
 \beta_{j,+}(z)=-\beta_{j,-}(z),\qquad z\in\Gamma_j.
\end{equation}
\end{proposition}

\begin{proof}
For $z$ in a compact subset of $\C\setminus\Gamma_j^{\rm cl}$,
\begin{align*}
 \log B_{\mathbf n,j}(z)
 &=\sum_{m=1}^{n_j}
 \log\!\left(1+\frac{z_{m,j}-w_{m,j}}{z-z_{m,j}}\right)\\
 &=\frac12\sum_{m=1}^{n_j}
 \frac{\ii\ee^{\ii\vartheta_{m,j}}h_j}
 {z-\ee^{\ii\vartheta_{m,j}}}+O(n_jh_j^2),
\end{align*}
where $\vartheta_{m,j}$ may be chosen between the two adjacent angular nodes.  Since $n_jh_j^2=O(n_j^{-1})$ and the sum is a Riemann sum,
\begin{equation*}
 \log B_{\mathbf n,j}(z)\longrightarrow
 \frac12\int_{\hat{a}_j}^{\hat{b}_j}
 \frac{\ii\ee^{\ii\vartheta}}{z-\ee^{\ii\vartheta}}\,\dd\vartheta
 =\frac12\log\frac{z-\eta_{2j-1}}{z-\eta_{2j}}.
\end{equation*}
Exponentiation proves \eqref{eq:B-limit}; \eqref{eq:beta-boundary} is the standard square-root boundary relation.
\end{proof}

\begin{proposition}[Thermodynamic convergence at the RHP level]
\label{prop:RHP-convergence}
On every fixed pole-removal loop, the jump matrices of $D^{(\mathbf n)}$ converge locally uniformly, and hence in every $L^p$, $1\le p\le\infty$, to the jumps obtained by replacing $B_{\mathbf n,j}$ with $\beta_j$.  If the limiting Beals--Coifman operator is invertible, then the corresponding RHP solutions converge locally uniformly away from the contours; in particular, the reconstructed potentials converge on compact subsets of the $(x,t)$-plane.
\end{proposition}

\begin{proof}
The jump convergence follows from Proposition~\ref{prop:B-limit}.  The Cauchy singular-integral operators therefore converge in operator norm on $L^2$ of the fixed loops.  Invertibility of the limiting operator is stable under a sufficiently small operator-norm perturbation, and the standard resolvent identity gives convergence of the boundary densities and of their Cauchy transforms.  The coefficient at infinity then converges as well.  This is the usual small-norm  theorem for normalized Riemann--Hilbert problems; see \cite{Zhou1989,DeiftZhou2002}.
\end{proof}

\subsection{Collapse of the loops and emergence of the full jump}

Passing to the limit in \eqref{eq:Dfinite} gives loop jumps
\begin{equation}\label{eq:F-continuum}
 F_j(z)=L\!\left(-a_j(z)\beta_j(z)\ee^{\ii\Phi(z;x,t)}\right)
 U\!\left(b_j(z)\beta_j(z)^{-1}\ee^{-\ii\Phi(z;x,t)}\right).
\end{equation}
Remove these loop jumps by setting, inside $\mathscr C_j$,
\begin{equation}\label{eq:collapse-transform}
 X(z):=D^{(\infty)}(z)F_j(z)^{-1}
 =D^{(\infty)}(z)
 U\!\left(-b_j\beta_j^{-1}\ee^{-\ii\Phi}\right)
 L\!\left(a_j\beta_j\ee^{\ii\Phi}\right),
\end{equation}
with the Schwarz-reflected definition inside $\mathscr C_j^\sharp$, and $X=D^{(\infty)}$ elsewhere.  Since $\beta_{j,-}=-\beta_{j,+}$, the jump remaining on $\Gamma_j$ is
\begin{equation}\label{eq:raw-LUL-jump}
 X_+(z)=X_-(z)
 L\!\left(a_j\beta_{j,+}\ee^{\ii\Phi}\right)
 U\!\left(-2b_j\beta_{j,+}^{-1}\ee^{-\ii\Phi}\right)
 L\!\left(a_j\beta_{j,+}\ee^{\ii\Phi}\right).
\end{equation}
This three-factor $LUL$ structure, rather than a product of only two triangular factors, is the essential consequence of the alternating discrete spectrum.

Define the final continuum densities on $\Gamma_j$ by
\begin{equation}\label{eq:raw-to-final-data}
 a_j(z)\beta_{j,+}(z)=-\ii r(z),
 \qquad
 b_j(z)\beta_{j,+}(z)^{-1}
 =\ii\frac{\rho(z)}{1+r(z)\rho(z)}.
\end{equation}
The second relation is meaningful whenever $1+r\rho\neq0$.  Substitution into \eqref{eq:raw-LUL-jump} gives
\begin{align}
 &L\!\left(-\ii r\ee^{\ii\Phi}\right)
 U\!\left(-\frac{2\ii\rho}{1+r\rho}\ee^{-\ii\Phi}\right)
 L\!\left(-\ii r\ee^{\ii\Phi}\right)\nonumber\\
 &\hspace{1cm}=
 \frac{1}{1+r\rho}
 \begin{pmatrix}
 1-r\rho&-2\ii\rho\ee^{-\ii\Phi}\\
 -2\ii r\ee^{\ii\Phi}&1-r\rho
 \end{pmatrix}.
 \label{eq:LUL-full-matrix}
\end{align}

\subsection{Continuum reduction data and the full gas RHP}

For lower-arc formulas, define
\begin{equation}\label{eq:sharp-data}
 r^\sharp(z):=\ol{r(\ol z)},\qquad
 \rho^\sharp(z):=\ol{\rho(\ol z)},
 \qquad z\in\bigcup_{j=1}^{N}\ol\Gamma_j.
\end{equation}
The Schwarz reduction gives the lower jump by conjugating the upper one.  The reciprocal reduction gives the same lower jump precisely when
\begin{equation}\label{eq:continuum-reduction-condition}
 r^\sharp(z)=r(z^{-1}),\qquad
 \rho^\sharp(z)=\rho(z^{-1}),
 \qquad z\in\bigcup_{j=1}^{N}\ol\Gamma_j.
\end{equation}
Because $z^{-1}=\ol z$ on the unit circle, \eqref{eq:continuum-reduction-condition} is equivalent to
\begin{equation}\label{eq:real-upper-data}
 r(z),\rho(z)\in\R,
 \qquad z\in\bigcup_{j=1}^{N}\Gamma_j.
\end{equation}
Thus explicit complex conjugation is required in the lower jump unless real-valued upper data have already been imposed.

\begin{assumption}[Minimal continuum data]\label{ass:data}
For every $j=1,\ldots,N$, the functions $r$ and $\rho$ extend analytically to a neighborhood of $\Gamma_j$, are real-valued on $\Gamma_j$, and satisfy
\begin{equation}\label{eq:minimal-data}
 1+r(z)\rho(z)\neq0,
 \qquad z\in\Gamma_j.
\end{equation}
For a genuinely full gas we additionally assume $r(z)\rho(z)\neq0$ on every spectral arc.  No positivity assumption is made in this definition.
\end{assumption}

\begin{RHP}[Full genus-$N$ dark soliton gas]\label{rhp:full}
Find a $2\times2$ matrix $M(z)=M(z;x,t)$ with the following properties.
\begin{enumerate}[label=(\roman*)]
\item $M$ is analytic for
\begin{equation*}
 z\in\C\setminus\left(\{0\}\cup\bigcup_{j=1}^N(\Gamma_j\cup\ol\Gamma_j)\right).
\end{equation*}
At every endpoint $\eta_m^{\pm1}$, both $M$ and $M^{-1}$ have at most fourth-root growth.  Equivalently, the RHP is understood in the standard $L^2$ endpoint class.
\item As $z\to\infty$ and $z\to0$,
\begin{equation}\label{eq:M-normalization}
 M(z)=I+O(z^{-1}),\qquad
 M(z)=\frac{\sigma_1}{z}+O(1),
\end{equation}
respectively.
\item The strong symmetries are
\begin{equation}\label{eq:M-symmetry}
 M(z)=\sigma_1\ol{M(\ol z)}\sigma_1
 =z^{-1}M(z^{-1})\sigma_1.
\end{equation}
\item The jump relation is $M_+=M_-J_M$.  On the upper arcs,
\begin{equation}\label{eq:upper-jump}
 J_M(z)=\frac{1}{1+r(z)\rho(z)}
 \begin{pmatrix}
 1-r(z)\rho(z)&-2\ii\rho(z)\ee^{-\ii\Phi(z;x,t)}\\
 -2\ii r(z)\ee^{\ii\Phi(z;x,t)}&1-r(z)\rho(z)
 \end{pmatrix},
 \quad z\in\bigcup_{j=1}^N\Gamma_j.
\end{equation}
On the lower arcs,
\begin{equation}\label{eq:lower-jump}
 J_M(z)=\frac{1}{1+r^\sharp(z)\rho^\sharp(z)}
 \begin{pmatrix}
 1-r^\sharp(z)\rho^\sharp(z)&2\ii r^\sharp(z)\ee^{-\ii\Phi(z;x,t)}\\
 2\ii\rho^\sharp(z)\ee^{\ii\Phi(z;x,t)}&1-r^\sharp(z)\rho^\sharp(z)
 \end{pmatrix},
 \quad z\in\bigcup_{j=1}^N\ol\Gamma_j.
\end{equation}
Under Assumption~\ref{ass:data}, one may equivalently replace $r^\sharp(z),\rho^\sharp(z)$ by $r(z^{-1}),\rho(z^{-1})$.
\end{enumerate}
\end{RHP}

Every jump matrix has determinant one.  The potential is reconstructed by
\begin{equation}\label{eq:reconstruction}
 q(x,t)=\lim_{z\to\infty}zM_{21}(z;x,t).
\end{equation}

\begin{remark}[Positivity versus the minimal RHP condition]\label{rem:positivity}
For the algebraic construction of RHP~\ref{rhp:full}, the only denominator condition is $1+r\rho\neq0$.  Individual nonvanishing of $r$ and $\rho$ is needed only if one wants both triangular factorizations everywhere.  Positivity of $r$ and $\rho$ is therefore not part of the definition of the full gas.  The nonlinear steepest-descent proof below uses a stronger admissibility condition to choose single-valued logarithms, apply a convenient vanishing lemma, and keep the parabolic-cylinder parameter real.  Positivity is a simple sufficient condition for those later steps, not for the continuum-limit RHP itself.
\end{remark}

\begin{assumption}[Steepest-descent admissibility]\label{ass:steepest}
In addition to Assumption~\ref{ass:data}, suppose that
\begin{equation}\label{eq:positive-product-data}
 r(z)>0,\qquad \rho(z)>0,\qquad 0<r(z)\rho(z)<1
\end{equation}
on every upper spectral arc.  In particular,
\begin{equation}\label{eq:crcrho-positive}
 c_r(z):=\frac{2r(z)}{1+r(z)\rho(z)}>0,
 \qquad
 c_\rho(z):=\frac{2\rho(z)}{1+r(z)\rho(z)}>0,
\end{equation}
and
\begin{equation}\label{eq:crcrho-product}
 0<c_r(z)c_\rho(z)<1.
\end{equation}
Both functions are further assumed to have nonvanishing analytic continuations to the lens neighborhoods, and that the chosen logarithms have zero winding there.  This is the sufficient hypothesis used in the long-time theorem.  More general real or complex zero-index data require a corresponding modification of the scalar and local model problems.
\end{assumption}

\begin{proposition}[Well-posedness under the admissibility hypothesis]
\label{prop:solvability}
Under Assumption~\ref{ass:steepest}, RHP~\ref{rhp:full} has a unique solution for every $(x,t)\in\mathbb R^2$.  The solution depends smoothly on $x,t$, and the reconstruction formula \eqref{eq:reconstruction} produces a solution of \eqref{eq:dnls}.
\end{proposition}

\begin{proof}
We record the coercive reduction because the full jump is not triangular.  On an upper arc put $s=r\rho>0$.  Since $z=\ee^{\ii\vartheta}$ implies
$\Phi(z;x,t)=2\ii[x\sin\vartheta-t\sin(2\vartheta)]$, the quantity
$\ee^{-2\ii\Phi}$ is positive there.  Hence
\[
 d(z)=\left(\frac{\rho(z)}{r(z)}\ee^{-2\ii\Phi(z;x,t)}\right)^{1/4}>0,
 \qquad D(z)=\diag(d(z),d(z)^{-1}).
\]
The chosen zero-index analytic continuations make this gauge single valued in a neighborhood of each arc.  A direct calculation gives
\begin{equation}\label{eq:unitary-gauge}
 D^{-1}J_MD=
 \begin{pmatrix}
 \dfrac{1-s}{1+s}&-\ii\dfrac{2\sqrt{s}}{1+s}\\[2mm]
 -\ii\dfrac{2\sqrt{s}}{1+s}&\dfrac{1-s}{1+s}
 \end{pmatrix},
\end{equation}
which is unitary because the squares of the two real coefficients sum to one.  Moreover, by $0<s<1$ its Hermitian part is strictly positive:
\[
 (D^{-1}J_MD)+(D^{-1}J_MD)^*
 =2\frac{1-s}{1+s}I>0.
\]
The lower-arc jump has the reflected unitary reduction.  The associated Beals--Coifman operator is Fredholm of index zero: the jumps are continuously homotopic to the identity by replacing $(r,\rho)$ with $(\lambda r,\lambda\rho)$, $0\le\lambda\le1$, without crossing a singular jump.  The positive-Hermitian-part vanishing lemma, applied after the Joukowski folding (or equivalently by Schwarz reflection across the unit circle), shows that the homogeneous problem has only the zero solution.  The Fredholm alternative yields existence, and the determinant/Liouville argument yields uniqueness.  This is the standard vanishing-lemma mechanism of \cite{Zhou1989}; the finite-density normalization at $z=0$ is handled exactly as in the dark-soliton RHP of \cite{CuccagnaJenkins2016,BertolaWangYanZhu2026}.  Differentiating the RHP with respect to $x,t$ and using Liouville's theorem gives the dNLS Lax pair and the reconstruction statement.
\end{proof}

Combining Propositions~\ref{prop:RHP-convergence} and~\ref{prop:solvability} shows that the exact finite-pole precursors converge to the full gas on compact $(x,t)$ sets whenever the discrete data satisfy the above admissibility conditions for all sufficiently large $\mathbf n$.

\section{The fixed genus-\texorpdfstring{$N$}{N} \texorpdfstring{$g$}{g}-function}\label{sec:g}

Set
\begin{equation}\label{eq:xi-def}
 \xi=\frac{x}{t},\qquad
 \Phi(z;x,t)=t\theta(z;\xi),\qquad
 \theta(z;\xi)=\xi(z-z^{-1})-(z^2-z^{-2}).
\end{equation}
Define
\begin{equation}\label{eq:RN}
 \RN(z)=\prod_{m=1}^{2N}\sqrt{(z-\eta_m)(z-\eta_m^{-1})},
\end{equation}
with cuts on all $2N$ spectral arcs, and choose the branch by
\begin{equation}\label{eq:RN-normalization}
 \RN(z)\sim z^{2N}\quad(z\to\infty),\qquad \RN(0)=1.
\end{equation}

Let $d\phi_N^{(1)}$ and $d\phi_N^{(2)}$ be the unique residue-free, real-normalized second-kind differentials invariant under the reciprocal symmetry and satisfying
\begin{align}
 d\phi_N^{(1)}(z)&=\frac{\PN(z)}{\RN(z)}\dd z,\label{eq:dphi1}\\
 d\phi_N^{(2)}(z)&=\frac{\QN(z)}{\RN(z)}\dd z.\label{eq:dphi2}
\end{align}
Their principal parts on the main sheet are
\begin{align}
 d\phi_N^{(1)}(z)&=(1+O(z^{-2}))\dd z,&&z\to\infty,
 &d\phi_N^{(1)}(z)&=(z^{-2}+O(1))\dd z,&&z\to0,\label{eq:dphi1-principal}\\
 d\phi_N^{(2)}(z)&=(z+O(z^{-2}))\dd z,&&z\to\infty,
 &d\phi_N^{(2)}(z)&=(z^{-3}+O(1))\dd z,&&z\to0.\label{eq:dphi2-principal}
\end{align}
The $N$ independent $a$-periods vanish.  Equivalently,
\begin{equation}\label{eq:PN-form}
 \PN(z)=\frac{\displaystyle
 \sum_{m=0}^{N}B_m z^{2N+2-m}+B_{N+1}z^{N+1}+\sum_{m=0}^{N}B_mz^m}{z^2},
\end{equation}
\begin{equation}\label{eq:QN-form}
 \QN(z)=\frac{\displaystyle
 \sum_{m=0}^{N+1}C_m z^{2N+4-m}+C_{N+2}z^{N+2}+\sum_{m=0}^{N+1}C_mz^m}{z^3},
\end{equation}
where the real coefficients are fixed by \eqref{eq:dphi1-principal}--\eqref{eq:dphi2-principal}, absence of residues, and the zero-period conditions.

For $a=1,2$, introduce the Abelian integrals
\begin{equation}\label{eq:phiN-components}
 \phi_N^{(a)}(z)=\int_{\eta_1}^{z}d\phi_N^{(a)},
\end{equation}
where the integration path is taken on the main sheet in the cut plane fixed below.  Define
\begin{equation}\label{eq:phiN}
 \phi_N(z;\xi)=\frac{\xi}{2}\phi_N^{(1)}(z)-\phi_N^{(2)}(z),
\end{equation}
and
\begin{equation}\label{eq:gN}
 g_N(z;\xi)=-\frac12\theta(z;\xi)+\phi_N(z;\xi).
\end{equation}
Then
\begin{align}
 g_N(z;\xi)&=g_{N,\infty}(\xi)+\frac{g_{N,1}(\xi)}{z}+O(z^{-2}),&&z\to\infty,\label{eq:gN-infty}\\
 g_N(z;\xi)&=-g_{N,\infty}(\xi)+O(z),&&z\to0.\label{eq:gN-zero}
\end{align}
On every upper or lower spectral arc,
\begin{equation}\label{eq:gN-band}
 g_{N,+}(z;\xi)+g_{N,-}(z;\xi)=-\theta(z;\xi).
\end{equation}
Define the upper gaps
\begin{equation}\label{eq:gaps}
 \mathfrak g_m=(\eta_{2m},\eta_{2m+1}),\quad m=1,\ldots,N-1,
 \qquad
 \mathfrak g_N=(\eta_{2N},-1),
\end{equation}
and their lower reflections.  On these gaps,
\begin{equation}\label{eq:gN-gap}
 g_{N,+}(z;\xi)-g_{N,-}(z;\xi)=\Omega_{N,m}(\xi),
 \qquad z\in\mathfrak g_m\cup\ol{\mathfrak g}_m.
\end{equation}
The constants in \eqref{eq:gN-gap} are the $b$-periods of the two normalized second-kind differentials.  More precisely, orient the canonical cycle $b_m$ so that its projection on the main sheet runs outside the unit circle from $\eta_1$ to $\eta_{2m}$ and returns on the second sheet.  We set
\begin{equation}\label{eq:OmegaNm-components}
 \Omega_{N,m}^{(a)}
 :=\oint_{b_m}d\phi_N^{(a)}
 =2\int_{\eta_1}^{\eta_{2m}}d\phi_N^{(a)},
 \qquad a=1,2,\quad m=1,\ldots,N,
\end{equation}
where the displayed open integral is taken along the exterior path on the main sheet.  The vanishing $a$-periods make this prescription independent of deformations of the path that do not cross a branch cut.  Consequently,
\begin{align}
 \Omega_{N,m}(\xi)
 &=\oint_{b_m}\left(\frac{\xi}{2}d\phi_N^{(1)}-d\phi_N^{(2)}\right)\label{eq:OmegaNm-period}\\
 &=\frac{\xi}{2}\Omega_{N,m}^{(1)}-\Omega_{N,m}^{(2)}\notag\\
 &=2\int_{\eta_1}^{\eta_{2m}}
 \left(\frac{\xi}{2}\frac{\PN(s)}{\RN(s)}-\frac{\QN(s)}{\RN(s)}\right)\dd s.
 \label{eq:OmegaNm-explicit}
\end{align}
All quantities in \eqref{eq:OmegaNm-components}--\eqref{eq:OmegaNm-explicit} are real because the differentials are real normalized.  With
\begin{equation}\label{eq:OmegaN-vectors}
 \Omega_N^{(a)}=
 \bigl(\Omega_{N,1}^{(a)},\ldots,\Omega_{N,N}^{(a)}\bigr)^T,
 \qquad a=1,2,
\end{equation}
we therefore have the explicit affine vector relation
\begin{equation}\label{eq:OmegaN-vector}
 \Omega_N(\xi)=\bigl(\Omega_{N,1}(\xi),\ldots,\Omega_{N,N}(\xi)\bigr)^T
 =\frac{\xi}{2}\Omega_N^{(1)}-\Omega_N^{(2)}.
\end{equation}
The reciprocal and Schwarz symmetries imply that the same constant $\Omega_{N,m}(\xi)$ occurs on the reflected lower gap.

The phase derivative is
\begin{equation}\label{eq:phiN-prime}
 \phi_N'(z;\xi)=\frac{\xi\PN(z)-2\QN(z)}{2\RN(z)}
 =\frac{\PN(z)}{2\RN(z)}\bigl(\xi-V_N(z)\bigr),
 \qquad V_N(z)=\frac{2\QN(z)}{\PN(z)}.
\end{equation}

\subsection{Quotient representation and zero distribution}\label{subsec:velocity-proof}

Put
\begin{equation}\label{eq:e-extended}
 e_0=1,\qquad e_{2N+1}=-1,
\end{equation}
and introduce the quotient radical
\begin{equation}\label{eq:Rscript}
 \mathscr R_N(k)^2=\prod_{m=0}^{2N+1}(k-e_m)
 =(k^2-1)\prod_{m=1}^{2N}(k-e_m),
 \qquad \mathscr R_N(k)\sim k^{N+1}\quad(k\to\infty).
\end{equation}
The $N+1$ real cuts of the quotient curve and the $N$ complementary spectral gaps are
\begin{equation}\label{eq:quotient-bands-gaps}
 \mathscr B_j=(e_{2j+1},e_{2j}),\quad j=0,\ldots,N,
 \qquad
 \mathscr G_\ell=(e_{2\ell},e_{2\ell-1}),\quad \ell=1,\ldots,N.
\end{equation}
Under the Joukowski map, $\Gamma_\ell$ is mapped bijectively onto $\mathscr G_\ell$, with $k=\cos\vartheta$ decreasing as $\vartheta$ increases.

\begin{lemma}[Descent of the normalized differentials]\label{lem:quotient-differentials}
There exist real polynomials $\mathsf P_N$ and $\mathsf Q_N$ such that
\begin{equation}\label{eq:PQ-pushdown}
 \PN(z)=2^N z^{N-1}\mathsf P_N(k(z)),
 \qquad
 \QN(z)=2^N z^{N-1}\mathsf Q_N(k(z)),
\end{equation}
and
\begin{equation}\label{eq:differential-pushdown}
 d\phi_N^{(1)}=\frac{\mathsf P_N(k)}{\mathscr R_N(k)}\,dk,
 \qquad
 d\phi_N^{(2)}=\frac{\mathsf Q_N(k)}{\mathscr R_N(k)}\,dk.
\end{equation}
Their degrees and leading coefficients are
\begin{equation}\label{eq:PQ-leading}
 \mathsf P_N(k)=2k^{N+1}+O(k^N),
 \qquad
 \mathsf Q_N(k)=4k^{N+2}+O(k^{N+1}).
\end{equation}
Moreover, for every $j=0,\ldots,N$,
\begin{equation}\label{eq:PQ-zero-periods-k}
 \int_{\mathscr B_j}\frac{\mathsf P_N(k)}{\mathscr R_{N,+}(k)}\,dk=0,
 \qquad
 \int_{\mathscr B_j}\frac{\mathsf Q_N(k)}{\mathscr R_{N,+}(k)}\,dk=0.
\end{equation}
The $N+1$ identities in each family contain one homological dependence.
\end{lemma}

\begin{proof}
Let $\mathscr S_N$ be the compactification of the reciprocal $z$-curve
\begin{equation*}
 W^2=\RN(z)^2=\prod_{m=1}^{2N}(z-\eta_m)(z-\eta_m^{-1}).
\end{equation*}
The involution whose quotient is the $k$-curve is
\begin{equation}\label{eq:quotient-involution}
 \iota(z,W)=\bigl(z^{-1},-z^{-2N}W\bigr).
\end{equation}
Indeed, the two functions
\begin{equation}\label{eq:quotient-map-rigorous}
 k=\frac12(z+z^{-1}),
 \qquad
 y=\frac{z-z^{-1}}{2^{N+1}z^N}W
\end{equation}
are invariant under $\iota$, and direct multiplication of reciprocal factors gives
\begin{equation}\label{eq:R-pushdown}
 y^2=(k^2-1)\prod_{m=1}^{2N}(k-e_m)=\mathscr R_N(k)^2,
 \qquad
 \mathscr R_N(k(z))
 =\frac{z-z^{-1}}{2^{N+1}z^N}\RN(z).
\end{equation}
The sign in the last identity is fixed at the point $z=\infty$, where both sides are asymptotic to $z^{N+1}/2^{N+1}$.

The palindromic forms \eqref{eq:PN-form}--\eqref{eq:QN-form} imply
\begin{equation}\label{eq:PQ-reciprocal-relations}
 \PN(z^{-1})=z^{2-2N}\PN(z),
 \qquad
 \QN(z^{-1})=z^{2-2N}\QN(z).
\end{equation}
Consequently the differentials are invariant under the quotient involution.  For example,
\begin{equation*}
 \iota^*\!\left(\frac{\PN(z)}{W}\,dz\right)
 =\frac{\PN(z^{-1})}{-z^{-2N}W}\,d(z^{-1})
 =\frac{\PN(z)}{W}\,dz,
\end{equation*}
and the same computation applies to $\QN$.  Equivalently,
$2^{-N}z^{1-N}\PN(z)$ and $2^{-N}z^{1-N}\QN(z)$ are finite Laurent polynomials invariant under $z\mapsto z^{-1}$.  Every such Laurent polynomial is a polynomial in $z+z^{-1}=2k$.  This defines the real polynomials $\mathsf P_N$ and $\mathsf Q_N$ and proves \eqref{eq:PQ-pushdown}.  Combining \eqref{eq:R-pushdown} with
\begin{equation*}
 dk=\frac{z-z^{-1}}{2z}\,dz
\end{equation*}
then gives \eqref{eq:differential-pushdown}.  The principal parts
\eqref{eq:dphi1-principal}--\eqref{eq:dphi2-principal} show that the degrees are $N+1$ and $N+2$ and give the leading coefficients in \eqref{eq:PQ-leading}.

It remains to identify the periods.  Let $\widetilde a_j$ be the positively oriented oval of the quotient surface surrounding the cut $\mathscr B_j$.  For $j=1,\ldots,N$, its lift is the reciprocal pair of normalized gap cycles on $\mathscr S_N$, and hence its period vanishes.  The remaining oval $\widetilde a_0$ is not independent: on the compact quotient surface,
\begin{equation*}
 \widetilde a_0+\widetilde a_1+\cdots+\widetilde a_N
\end{equation*}
is homologous to the union of small circles about $\infty_+$ and $\infty_-$.  Both descended differentials are of the second kind and therefore have zero residues at the two points at infinity, so the $\widetilde a_0$-period vanishes as well.  Finally,
\begin{equation*}
 \oint_{\widetilde a_j}\frac{\mathsf P_N(k)}{\mathscr R_N(k)}\,dk
 =2\int_{\mathscr B_j}\frac{\mathsf P_N(k)}{\mathscr R_{N,+}(k)}\,dk,
\end{equation*}
up to the common orientation sign, and analogously for $\mathsf Q_N$.  This proves \eqref{eq:PQ-zero-periods-k} and also explains the single homological dependence.
\end{proof}

The following quotient-curve zero count is a fixed-curve analogue of the real-normalized-differential argument used in Whitham theory; compare \cite{Levermore1988,BertolaWangYanZhu2026}.  The quotient velocity is
\begin{equation}\label{eq:Vk-quotient}
 \mathsf V_N(k)=\frac{2\mathsf Q_N(k)}{\mathsf P_N(k)},
 \qquad
 V_N(z)=\mathsf V_N(k(z)).
\end{equation}

\begin{lemma}[One zero on every quotient cut]\label{lem:band-zeros}
The polynomial $\mathsf P_N$ has exactly one simple zero $p_j$ in each
$\mathscr B_j$, $j=0,\ldots,N$.  Here and throughout Section~\ref{sec:g},
$p_j$ denotes a zero of the second-kind numerator $\mathsf P_N$; it is
unrelated to the Baker--Akhiezer divisor points introduced in
Section~\ref{subsec:BA-model}.  Let $k_0$ belong to the closure of one of the
spectral gaps $\mathscr G_\ell$ and set
\begin{equation}\label{eq:xi0-Vk0}
 \xi_0=\mathsf V_N(k_0),
 \qquad
 \mathsf G_{k_0}(k)=2\mathsf Q_N(k)-\xi_0\mathsf P_N(k).
\end{equation}
Then $\mathsf G_{k_0}$ has exactly one simple zero $s_j(k_0)$ in each
$\mathscr B_j$ and one additional simple zero at $k=k_0$.
\end{lemma}

\begin{proof}
For $k\in\mathscr B_j$, the boundary value of the radical has the form
\begin{equation}\label{eq:R-cut-phase}
 \mathscr R_{N,+}(k)
 =\ii\varepsilon_j
 \sqrt{\left|\prod_{m=0}^{2N+1}(k-e_m)\right|},
 \qquad \varepsilon_j\in\{1,-1\},
\end{equation}
with $\varepsilon_j$ constant on the whole open interval.  Therefore the
first identity in \eqref{eq:PQ-zero-periods-k} is equivalent, after
multiplication by the nonzero constant $\ii\varepsilon_j$, to
\begin{equation*}
 \int_{\mathscr B_j}\mathsf P_N(k)
 \frac{dk}{\sqrt{\left|\prod_{m=0}^{2N+1}(k-e_m)\right|}}=0.
\end{equation*}
The weight is strictly positive in the interior of $\mathscr B_j$.  Since
$\mathsf P_N$ is not identically zero, the vanishing integral forces it to
take both signs and hence to possess an interior zero in every cut.  The
$N+1$ cuts are disjoint and $\deg\mathsf P_N=N+1$, so these are all the
zeros of $\mathsf P_N$.  If one of them had multiplicity larger than one,
the total multiplicity would exceed $N+1$; hence every zero is simple.

The same period argument applies to
$\mathsf G_{k_0}/\mathscr R_N\,dk$, because its cut periods are the same
linear combination of the two families in
\eqref{eq:PQ-zero-periods-k}.  Thus $\mathsf G_{k_0}$ has an interior zero
in each of the $N+1$ cuts.  It also vanishes at $k_0$ by definition.  The
point $k_0$ lies in the closure of a spectral gap and therefore outside the
interior of every $\mathscr B_j$, so these $N+2$ zeros are distinct.  From
\eqref{eq:PQ-leading},
\begin{equation*}
 \mathsf G_{k_0}(k)=8k^{N+2}+O(k^{N+1}),
\end{equation*}
so $\deg\mathsf G_{k_0}=N+2$.  The $N+2$ distinct zeros therefore exhaust
the degree, and each of them, including $k_0$, is simple.
\end{proof}

The preceding zero count gives an exact factorization:
\begin{equation}\label{eq:velocity-factorization}
 \mathsf V_N(k)-\mathsf V_N(k_0)
 =4(k-k_0)\prod_{j=0}^{N}
 \frac{k-s_j(k_0)}{k-p_j}.
\end{equation}
Indeed, the numerator on the left has leading coefficient $8$, while $\mathsf P_N$ has leading coefficient $2$.  If $k$ is outside the interiors of all $\mathscr B_j$, then $p_j$ and $s_j(k_0)$ lie in the same interval $\mathscr B_j$, and therefore
\begin{equation}\label{eq:positive-ratio-factors}
 \frac{k-s_j(k_0)}{k-p_j}>0,
 \qquad j=0,\ldots,N.
\end{equation}

\begin{theorem}[Strict velocity ordering on the fixed full curve]\label{thm:velocity-ordering}
For every endpoint configuration satisfying \eqref{eq:e-intro}, the following statements hold.
\begin{enumerate}[label=(\roman*)]
\item The function $\mathsf V_N$ is real analytic and strictly increasing on each $\mathscr G_\ell$:
\begin{equation}\label{eq:Vk-increasing}
 \mathsf V_N'(k)>0,
 \qquad k\in\mathscr G_\ell,
 \quad \ell=1,\ldots,N.
\end{equation}
Equivalently,
\begin{equation}\label{eq:V-decreasing}
 \frac{\dd}{\dd\vartheta}V_N(\ee^{\ii\vartheta})<0,
 \qquad \ee^{\ii\vartheta}\in\Gamma_\ell.
\end{equation}
\item The endpoint velocities are strictly ordered:
\begin{equation}\label{eq:global-order}
 v_1>v_2>\cdots>v_{2N},
 \qquad v_m=V_N(\eta_m)=\mathsf V_N(e_m).
\end{equation}
\item Every sector \eqref{eq:P0-intro}--\eqref{eq:Ml-intro} is nonempty.  For every $\xi\in\mathcal M_\ell$ there is a unique $\alpha_\ell(\xi)\in\Gamma_\ell$ satisfying \eqref{eq:alpha-intro}; this switching point is simple.
\end{enumerate}
\end{theorem}

\begin{proof}
Because $\mathsf P_N$ has no zero in any spectral gap,
$\mathsf V_N=2\mathsf Q_N/\mathsf P_N$ is real analytic there.  Fix
$k_0\in\mathscr G_\ell$.  Dividing
\eqref{eq:velocity-factorization} by $k-k_0$ and letting $k\to k_0$ gives
\begin{equation}\label{eq:Vprime-product}
 \mathsf V_N'(k_0)=4\prod_{j=0}^{N}
 \frac{k_0-s_j(k_0)}{k_0-p_j}>0.
\end{equation}
Indeed, $p_j$ and $s_j(k_0)$ lie in the same open cut $\mathscr B_j$, while
$k_0$ lies outside that cut, so the numerator and denominator of every
factor have the same sign.  This proves \eqref{eq:Vk-increasing}.  Since
$k=\cos\vartheta$ and $0<\vartheta<\pi$,
\begin{equation*}
 \frac{\dd}{\dd\vartheta}V_N(\ee^{\ii\vartheta})
 =-\sin\vartheta\,\mathsf V_N'(\cos\vartheta)<0,
\end{equation*}
which proves \eqref{eq:V-decreasing}.

For the global ordering, take $1\le m<n\le2N$ and use
\eqref{eq:velocity-factorization} with $k_0=e_n$ and $k=e_m$.  Lemma~\ref{lem:band-zeros}
applies because $e_n$ belongs to the closure of a spectral gap.  Moreover,
$e_m$ is never an interior point of a quotient cut.  Since the two points
$p_j$ and $s_j(e_n)$ are interior to the same cut, both
$e_m-p_j$ and $e_m-s_j(e_n)$ are nonzero and have the same sign.  Hence
\begin{equation}\label{eq:velocity-difference-product}
 v_m-v_n
 =4(e_m-e_n)\prod_{j=0}^{N}
 \frac{e_m-s_j(e_n)}{e_m-p_j}>0,
\end{equation}
because $e_m>e_n$.  This proves the pairwise inequalities in
\eqref{eq:global-order}.

Every interval in \eqref{eq:P0-intro}--\eqref{eq:Ml-intro} is therefore
nonempty.  On $\Gamma_\ell$, the continuous function $V_N$ decreases
strictly from $v_{2\ell-1}$ to $v_{2\ell}$, so the intermediate value theorem
gives a unique $\alpha_\ell(\xi)$ for every
$\xi\in\mathcal M_\ell$.  Finally,
\begin{equation*}
 \frac{\dd}{\dd z}V_N(z)
 =\mathsf V_N'(k(z))\,k'(z),
 \qquad
 k'(z)=\frac12(1-z^{-2}),
\end{equation*}
and both factors are nonzero at an interior point of $\Gamma_\ell$.  Since
$\PN$ has no zero on a spectral arc, formula \eqref{eq:phiN-prime} shows that
$\phi_N'(z;\xi)$ has a simple zero at $z=\alpha_\ell(\xi)$.
\end{proof}

\section{Triangular factorizations and the full Szeg\H{o} function}\label{sec:factor}

Set
\begin{equation}\label{eq:c-functions}
 c_\rho(z)=\frac{2\rho(z)}{1+r(z)\rho(z)},
 \qquad
 c_r(z)=\frac{2r(z)}{1+r(z)\rho(z)}.
\end{equation}
The upper jump \eqref{eq:upper-jump} has the two exact factorizations
\begin{align}
 J_M&=L_\rho
 \begin{pmatrix}0&-\ii c_\rho\ee^{-\ii\Phi}\\-\ii c_\rho^{-1}\ee^{\ii\Phi}&0\end{pmatrix}L_\rho,
 \label{eq:rho-factor}\\
 L_\rho&=\begin{pmatrix}
 1&0\\[-1mm]
 -\ii\dfrac{r\rho-1}{2\rho}\ee^{\ii\Phi}&1
 \end{pmatrix},\label{eq:Lrho-def}
\end{align}
and
\begin{align}
 J_M&=U_r
 \begin{pmatrix}0&-\ii c_r^{-1}\ee^{-\ii\Phi}\\-\ii c_r\ee^{\ii\Phi}&0\end{pmatrix}U_r,
 \label{eq:r-factor}\\
 U_r&=\begin{pmatrix}
 1&-\ii\dfrac{r\rho-1}{2r}\ee^{-\ii\Phi}\\[-1mm]
 0&1
 \end{pmatrix}.\label{eq:Ur-def}
\end{align}

For any fixed noncritical $\xi$, decompose every upper band into
\begin{equation}\label{eq:band-decomp}
 \Gamma_j^r(\xi)=\{z\in\Gamma_j:V_N(z)>\xi\},
 \qquad
 \Gamma_j^\rho(\xi)=\{z\in\Gamma_j:V_N(z)<\xi\}.
\end{equation}
In the pure sector $\mathcal P_\ell$,
\begin{equation}\label{eq:pure-decomp}
 \Gamma_j^r=\Gamma_j\ (j\le\ell),
 \qquad
 \Gamma_j^\rho=\Gamma_j\ (j>\ell).
\end{equation}
In the mixed sector $\mathcal M_\ell$,
\begin{equation}\label{eq:mixed-decomp}
 \begin{aligned}
 &\Gamma_j^r=\Gamma_j,&&j<\ell,\\
 &\Gamma_\ell^r=(\eta_{2\ell-1},\alpha_\ell),\qquad
 \Gamma_\ell^\rho=(\alpha_\ell,\eta_{2\ell}),\\
 &\Gamma_j^\rho=\Gamma_j,&&j>\ell.
 \end{aligned}
\end{equation}
Thus the active arc changes factorization once, while all later arcs remain present and carry $\rho$-lenses.

Define the scalar logarithmic datum
\begin{equation}\label{eq:ellxi}
 \ell_\xi(z)=
 \begin{cases}
 -\log c_r(z),&z\in\displaystyle\bigcup_{j=1}^N\Gamma_j^r(\xi),\\[1mm]
 \log c_\rho(z),&z\in\displaystyle\bigcup_{j=1}^N\Gamma_j^\rho(\xi).
 \end{cases}
\end{equation}
For $z$ on a lower arc, set
\begin{equation}\label{eq:ell-sharp}
 \ell_\xi^\sharp(z):=\ol{\ell_\xi(\ol z)}.
\end{equation}
Under the real dNLS reduction this is the same datum transported by $z\mapsto z^{-1}$.

\begin{definition}[Full genus-$N$ Szeg\H{o} function]\label{def:deltaN}
The scalar function $\delta_N(z;\xi)$ is analytic and nonzero in the complement of all bands and gaps and, in a mixed sector, away from the switching points $\alpha_\ell^{\pm1}$.  At a switching point it has the standard bounded power behavior used in the parabolic-cylinder model.  It satisfies
\begin{equation}\label{eq:delta-sym}
 \delta_N(z;\xi)\delta_N(z^{-1};\xi)=1,
 \qquad
 \delta_N(z;\xi)=\delta_{N,\infty}(\xi)+O(z^{-1}),\quad z\to\infty,
\end{equation}
and has jumps
\begin{align}
 \delta_{N,+}(z)\delta_{N,-}(z)&=\ee^{\ell_\xi(z)},
 &&z\in\bigcup_{j=1}^N\Gamma_j,\label{eq:delta-upper}\\
 \delta_{N,+}(z)\delta_{N,-}(z)&=\ee^{-\ell_\xi^\sharp(z)},
 &&z\in\bigcup_{j=1}^N\ol\Gamma_j,\label{eq:delta-lower}\\
 \frac{\delta_{N,+}(z)}{\delta_{N,-}(z)}&=\ee^{\ii\Delta_{N,m}(\xi)},
 &&z\in\mathfrak g_m\cup\ol{\mathfrak g}_m,
 \quad m=1,\ldots,N.\label{eq:delta-gap}
\end{align}
The constants $\Delta_{N,m}$ are fixed by the $N$ zero $a$-period conditions.  We write
\begin{equation}\label{eq:DeltaN}
 \Delta_N(\xi)=\bigl(\Delta_{N,1}(\xi),\ldots,\Delta_{N,N}(\xi)\bigr)^T.
\end{equation}
\end{definition}

An explicit Cauchy representation is
\begin{align}
 \delta_N(z;\xi)=\exp\Bigg\{\frac{\RN(z)}{2\pi\ii}\Bigg[&
 \sum_{j=1}^{N}\int_{\Gamma_j}
 \frac{\ell_\xi(s)}{R_{N,+}(s)(s-z)}\dd s
 -\sum_{j=1}^{N}\int_{\ol\Gamma_j}
 \frac{\ell_\xi^\sharp(s)}{R_{N,+}(s)(s-z)}\dd s\nonumber\\
 &+\sum_{m=1}^{N}\int_{\mathfrak g_m\cup\ol{\mathfrak g}_m}
 \frac{\ii\Delta_{N,m}(\xi)}{\RN(s)(s-z)}\dd s\Bigg]\Bigg\}.
 \label{eq:delta-Cauchy}
\end{align}
Here $R_{N,+}$ denotes the boundary value from the left side of each oriented band, both in the upper and lower half-planes; on the gaps, $R_N$ is single-valued in the chosen cut plane.
The moment equations eliminating the polynomial part at infinity are, for $p=0,\ldots,N-1$,
\begin{align}
 0={}&\sum_{j=1}^{N}\int_{\Gamma_j}
 \frac{\ell_\xi(s)s^p}{R_{N,+}(s)}\dd s
 -\sum_{j=1}^{N}\int_{\ol\Gamma_j}
 \frac{\ell_\xi^\sharp(s)s^p}{R_{N,+}(s)}\dd s
 +\sum_{m=1}^{N}\int_{\mathfrak g_m\cup\ol{\mathfrak g}_m}
 \frac{\ii\Delta_{N,m}(\xi)s^p}{\RN(s)}\dd s.
 \label{eq:delta-moments}
\end{align}
In a pure sector $\mathcal P_\ell$, the data $\Delta_N$ and $\delta_{N,\infty}$ are independent of $\xi$.  In a mixed sector $\mathcal M_\ell$, they depend on $\xi$ only through $\alpha_\ell(\xi)$.

\begin{proposition}[Existence and uniqueness of the full Szeg\H{o} factor]
\label{prop:Szego}
Under Assumption~\ref{ass:steepest}, the linear system \eqref{eq:delta-moments} has a unique real solution $\Delta_N(\xi)$.  Formula \eqref{eq:delta-Cauchy} then defines the unique scalar solution of Definition~\ref{def:deltaN} with the prescribed normalization and reciprocal symmetry.  In a mixed sector $\mathcal M_\ell$, locally in each component of a punctured neighborhood of $\alpha_\ell$,
\begin{equation}\label{eq:delta-local-switch}
 \delta_N(z;\xi)=\delta_{\ell,0}(z;\xi)
 (z-\alpha_\ell)^{\ii\nu_\ell\varepsilon},
 \qquad \delta_{\ell,0}(\alpha_\ell;\xi)\neq0,
\end{equation}
where $\varepsilon=\pm1$ depends only on the local sector and $\nu_\ell$ is given in \eqref{eq:nu}.
\end{proposition}

\begin{proof}
After pairing each upper gap with its reciprocal lower gap, the coefficient matrix of the unknowns $\Delta_{N,m}$ in \eqref{eq:delta-moments} is, up to nonzero orientation constants, the $a$-period matrix of the $N$ reciprocal-invariant holomorphic differentials on the quotient curve.  These differentials form a basis of the genus-$N$ holomorphic differentials.  If a linear combination of the columns vanished, the descended holomorphic differential would have all quotient $a$-periods zero and hence would vanish identically.  The matrix is therefore nonsingular.  Schwarz and reciprocal symmetry of the right-hand side imply that the solution is real.  The Sokhotski--Plemelj formulas give \eqref{eq:delta-upper}--\eqref{eq:delta-gap}; the moment conditions remove the polynomial part at infinity, and the reciprocal normalization fixes the remaining multiplicative constant.  Since \eqref{eq:delta-Cauchy} is an exponential, it has no zeros.  Finally, in a mixed sector the datum $\ell_\xi$ has one jump at $\alpha_\ell$.  Separating its constant one-sided values in the Cauchy integral gives the logarithmic term in \eqref{eq:delta-local-switch}; its coefficient is
$-\frac{1}{2\pi}\log(c_r(\alpha_\ell)c_\rho(\alpha_\ell))=\nu_\ell$.  This is the standard local Szeg\H{o} reduction used in full-jump stationary-point parametrices; compare \cite{YanGengWei2026,FullCH2026}.
\end{proof}

\section{All-band opening of lenses}\label{sec:lens}

This section corrects the key geometric distinction between a half gas and a full gas.  In the half gas, the genus-$\ell$ $g$-function only produces a leading model on the first $\ell$ populated arcs.  Here the $g$-function is built on the fixed full curve, and every one of the $N$ arcs must be opened and retained in the reduced model.

\subsection{Post-\texorpdfstring{$g$}{g} and post-Szeg\H{o} factorization}

Introduce
\begin{equation}\label{eq:S-transform}
 S(z)=\delta_{N,\infty}^{-\sigma_3}\ee^{-\ii t g_{N,\infty}\sigma_3}
 M(z)\ee^{\ii t g_N(z;\xi)\sigma_3}\delta_N(z;\xi)^{\sigma_3}.
\end{equation}
Since $\theta+2g_N=2\phi_N$, the jump on an upper $r$-portion is
\begin{equation}\label{eq:S-r-factor}
 J_S(z)=\mathcal U_{r,-}(z)(-\ii\sigma_1)\mathcal U_{r,+}(z),
\end{equation}
where
\begin{equation}\label{eq:U-pm}
 \mathcal U_{r,\pm}(z)=
 \begin{pmatrix}
 1&\displaystyle \ii\frac{1-r(z)\rho(z)}{2r(z)}
 \delta_{N,\pm}(z)^{-2}\ee^{-2\ii t\phi_{N,\pm}(z;\xi)}\\[2mm]
 0&1
 \end{pmatrix}.
\end{equation}
On an upper $\rho$-portion,
\begin{equation}\label{eq:S-rho-factor}
 J_S(z)=\mathcal L_{\rho,-}(z)(-\ii\sigma_1)\mathcal L_{\rho,+}(z),
\end{equation}
where
\begin{equation}\label{eq:L-pm}
 \mathcal L_{\rho,\pm}(z)=
 \begin{pmatrix}
 1&0\\[1mm]
 \displaystyle \ii\frac{1-r(z)\rho(z)}{2\rho(z)}
 \delta_{N,\pm}(z)^2\ee^{2\ii t\phi_{N,\pm}(z;\xi)}&1
 \end{pmatrix}.
\end{equation}
The lower factors are obtained from
\begin{equation}\label{eq:lower-factors}
 \mathcal A^{\rm low}_{\pm}(z)=
 \sigma_1\ol{\mathcal A_{\mp}(\ol z)}\sigma_1,
 \qquad \mathcal A\in\{\mathcal U_r,\mathcal L_\rho\},
\end{equation}
and their central jump is $+\ii\sigma_1$.

\subsection{The fixed-genus sign table}

Orient every upper band counterclockwise and define
\begin{equation}\label{eq:kappa}
 \kappa_j(\ee^{\ii\vartheta})=
 \frac{\ii\ee^{\ii\vartheta}\PN(\ee^{\ii\vartheta})}
 {2R_{N,+}(\ee^{\ii\vartheta})}.
\end{equation}

\begin{lemma}[Positivity of the quasimomentum density]\label{lem:kappa-positive}
For every $j=1,\ldots,N$,
\begin{equation}\label{eq:kappa-positive-claim}
 \kappa_j(\ee^{\ii\vartheta})>0,
 \qquad \ee^{\ii\vartheta}\in\Gamma_j.
\end{equation}
\end{lemma}

\begin{proof}
For a counterclockwise oriented upper arc, the plus side is the side
$|z|<1$.  Under the quotient map this side corresponds to the branch
$\mathscr R_N^{\rm in}$ for which, by \eqref{eq:R-pushdown} and
$\RN(0)=1$,
\begin{equation}\label{eq:Rin-infinity}
 \mathscr R_N^{\rm in}(k)\sim-k^{N+1}
\end{equation}
at the point at infinity represented by $z\to0$.  Let
$k=\cos\vartheta\in\mathscr G_j$.  Continuing from $k>1$ to
$\mathscr G_j$ along the real axis crosses precisely the cuts
$\mathscr B_0,\ldots,\mathscr B_{j-1}$, and the radical changes sign at
each crossing.  Hence
\begin{equation}\label{eq:Rin-sign}
 \operatorname{sgn}\mathscr R_N^{\rm in}(k)=(-1)^{j-1},
 \qquad k\in\mathscr G_j.
\end{equation}

The polynomial $\mathsf P_N$ has positive leading coefficient and, by
Lemma~\ref{lem:band-zeros}, exactly one simple zero in every cut.  Moving
from $+\infty$ to $\mathscr G_j$ crosses exactly the zeros
$p_0,\ldots,p_{j-1}$ and no others.  Therefore
\begin{equation}\label{eq:PN-gap-sign}
 \operatorname{sgn}\mathsf P_N(k)=(-1)^j,
 \qquad k\in\mathscr G_j.
\end{equation}
Using \eqref{eq:PQ-pushdown} and \eqref{eq:R-pushdown}, and recalling that
$z=\ee^{\ii\vartheta}$, we obtain
\begin{equation}\label{eq:kappa-k-form}
 \kappa_j(\ee^{\ii\vartheta})
 =-\frac{\sin\vartheta}{2}
 \frac{\mathsf P_N(k)}{\mathscr R_N^{\rm in}(k)}.
\end{equation}
The ratio in \eqref{eq:kappa-k-form} is negative by
\eqref{eq:Rin-sign}--\eqref{eq:PN-gap-sign}, while
$\sin\vartheta>0$ on the upper semicircle.  Thus
$\kappa_j>0$ on the whole open band.
\end{proof}

From \eqref{eq:phiN-prime},
\begin{equation}\label{eq:tangent-sign}
 \frac{\dd}{\dd\vartheta}\phi_{N,+}(\ee^{\ii\vartheta};\xi)
 =\kappa_j(\ee^{\ii\vartheta})
 \bigl(\xi-V_N(\ee^{\ii\vartheta})\bigr).
\end{equation}
The next proposition converts this tangential identity into the two-sided sign distribution needed for opening lenses.

\begin{proposition}[Sign distribution on all bands]\label{prop:signs}
Let $\xi$ stay in a compact subset of one open sector.
\begin{enumerate}[label=(\roman*)]
\item If $\xi\in\mathcal P_\ell$, then
\begin{equation}\label{eq:sign-pure}
 \Ima\phi_N<0\quad\text{adjacent to }\Gamma_j\ (j\le\ell),
 \qquad
 \Ima\phi_N>0\quad\text{adjacent to }\Gamma_j\ (j>\ell).
\end{equation}
\item If $\xi\in\mathcal M_\ell$, then
\begin{equation}\label{eq:sign-mixed}
 \begin{aligned}
 &\Ima\phi_N<0&&\text{adjacent to }\Gamma_j,\ j<\ell,\\
 &\Ima\phi_N<0&&\text{adjacent to }(\eta_{2\ell-1},\alpha_\ell),\\
 &\Ima\phi_N>0&&\text{adjacent to }(\alpha_\ell,\eta_{2\ell}),\\
 &\Ima\phi_N>0&&\text{adjacent to }\Gamma_j,\ j>\ell.
 \end{aligned}
\end{equation}
\end{enumerate}
Hence $\mathcal U_{r,\pm}-I$ and $\mathcal L_{\rho,\pm}-I$ are exponentially small on the corresponding lens lips, away from endpoints and, in a mixed sector, away from $\alpha_\ell$.
\end{proposition}

\begin{proof}
The velocity inequalities follow first from Theorem~\ref{thm:velocity-ordering}.
If $\xi\in\mathcal P_\ell$, then the whole range of $V_N$ on every
$\Gamma_j$ with $j\le\ell$ lies above $\xi$, whereas the whole range on
every $\Gamma_j$ with $j>\ell$ lies below $\xi$.  If
$\xi\in\mathcal M_\ell$, the same statement holds on the earlier and later
bands, and strict monotonicity gives
\begin{equation*}
 V_N>\xi\quad\hbox{on }(\eta_{2\ell-1},\alpha_\ell),
 \qquad
 V_N<\xi\quad\hbox{on }(\alpha_\ell,\eta_{2\ell}).
\end{equation*}

We now convert these inequalities into a two-sided sign statement.  The
real normalization of the second-kind differentials implies that
$\phi_{N,+}$ is real on every open upper band; equivalently this follows by
integrating \eqref{eq:tangent-sign}, whose right-hand side is real, from a
branch point and using that all intervening periods are real.  Parameterize
the band counterclockwise by arc length $s$, let $\mathbf t$ be its positive
unit tangent, and let $\mathbf n_+$ be the unit normal pointing into the plus
side $|z|<1$.  The Cauchy--Riemann equations give
\begin{equation}\label{eq:normal-CR-plus}
 \partial_{\mathbf n_+}\Ima\phi_{N,+}
 =\partial_s\Rea\phi_{N,+}
 =\kappa_j\bigl(\xi-V_N\bigr).
\end{equation}
On the minus side, the inward normal is $\mathbf n_-=-\mathbf n_+$.  The two boundary values satisfy $\phi_{N,-}=-\phi_{N,+}+c_j$ with a real band constant $c_j$; hence their imaginary parts and tangential derivatives have the required opposite relation, and the same normal formula follows,
\begin{equation}\label{eq:normal-CR-minus}
 \partial_{\mathbf n_-}\Ima\phi_{N,-}
 =\kappa_j\bigl(\xi-V_N\bigr).
\end{equation}
Consequently, uniformly on a closed subarc that avoids the endpoints and,
in a mixed sector, the switching point,
\begin{equation}\label{eq:normal-sign-expansion}
 \Ima\phi_N(z+h\mathbf n_\pm;\xi)
 =h\,\kappa_j(z)\bigl(\xi-V_N(z)\bigr)+O(h^2),
 \qquad h\downarrow0.
\end{equation}
Lemma~\ref{lem:kappa-positive} shows that the sign in
\eqref{eq:normal-sign-expansion} is precisely the sign of $\xi-V_N$ on both
sides of the band.  Combining this with the velocity inequalities proves
\eqref{eq:sign-pure} and \eqref{eq:sign-mixed}.

For $\xi$ in a compact subsector, the quantities
$|\xi-V_N|$ and $\kappa_j$ are bounded away from zero on closed subarcs
outside the local discs.  The lens lips can therefore be chosen inside the
strict sign neighborhoods furnished by \eqref{eq:normal-sign-expansion}.
It follows that $\mathcal U_{r,\pm}-I$ and
$\mathcal L_{\rho,\pm}-I$ are $O(\ee^{-ct})$ there for some $c>0$.
\end{proof}

\begin{table}[htbp]
\centering
\small
\renewcommand{\arraystretch}{1.35}
\begin{tabular}{c|c|c|c}
\toprule
sector & bands $j<\ell$ & active band $j=\ell$ & bands $j>\ell$\\
\midrule
$\mathcal P_\ell$ & $r$ on all & $r$ on all & $\rho$ on all\\
$\mathcal M_\ell$ & $r$ on all & $r$ on $(\eta_{2\ell-1},\alpha_\ell)$; $\rho$ on $(\alpha_\ell,\eta_{2\ell})$ & $\rho$ on all\\
\bottomrule
\end{tabular}
\caption{Fixed-genus-$N$ sign and factorization table.  The entries for $\mathcal P_0$ and $\mathcal P_N$ are understood by omitting the empty groups.  The $r$-factorization carries $\ee^{-2\ii t\phi_N}$ and is used where $\Ima\phi_N<0$; the $\rho$-factorization carries $\ee^{2\ii t\phi_N}$ and is used where $\Ima\phi_N>0$.}
\label{tab:signs}
\end{table}

\subsection{Crossed opening on the active band and full opening on the remaining bands}

For a pure sector, open an ordinary two-sided lens around every $\Gamma_j$, using $\mathcal U_r$ for $j\le\ell$ and $\mathcal L_\rho$ for $j>\ell$.  For a mixed sector $\mathcal M_\ell$, use full $r$-lenses around $\Gamma_j$, $j<\ell$, and full $\rho$-lenses around $\Gamma_j$, $j>\ell$.  Split the active band as in \eqref{eq:mixed-decomp}.  At $\alpha_\ell$, join the outer lip of the $r$-lens to the inner lip of the $\rho$-lens and the inner lip of the $r$-lens to the outer lip of the $\rho$-lens.  Thus the two lens systems cross at the switching point.  This is the unit-circle analogue of the crossed opening in the intermediate sector of a full CH or full KdV soliton gas.

Let $D_{j,\pm}^r$ and $D_{j,\pm}^\rho$ denote the corresponding lens domains adjacent to the plus and minus sides of the oriented upper bands.  Define
\begin{equation}\label{eq:T-transform}
 T(z)=S(z)
 \begin{cases}
 \mathcal U_{r,+}(z)^{-1},&z\in D_{j,+}^r,\\
 \mathcal U_{r,-}(z),&z\in D_{j,-}^r,\\
 \mathcal L_{\rho,+}(z)^{-1},&z\in D_{j,+}^\rho,\\
 \mathcal L_{\rho,-}(z),&z\in D_{j,-}^\rho,\\
 \text{the reflected factors from \eqref{eq:lower-factors}},&z\text{ in a lower lens},\\
 I,&\text{elsewhere}.
 \end{cases}
\end{equation}

Figure~\ref{fig:all-band-lenses} is drawn for a representative case $N=5$, $\ell=3$.  It includes the full lenses on the later arcs $j>\ell$, which were absent from the half-gas Figure~8 construction.

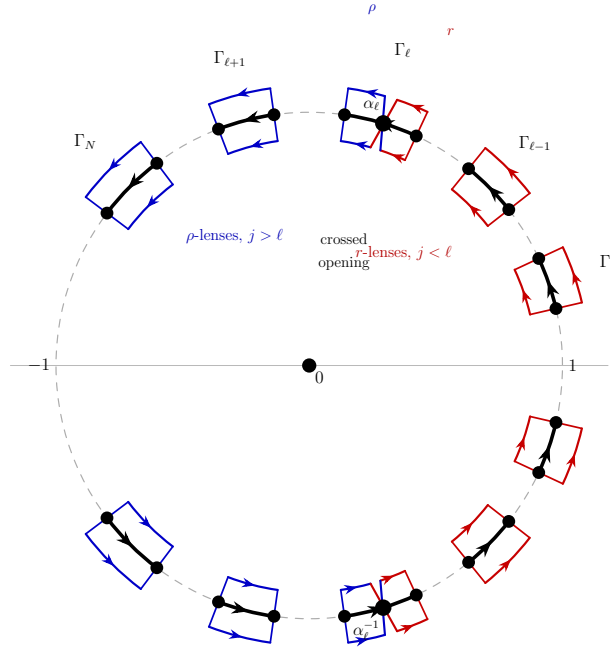
\begin{figure}[htbp]
\centering
\begin{tikzpicture}[scale=3.35,>=Stealth,line cap=round,line join=round,
 band/.style={black,line width=1.35pt,postaction={decorate},decoration={markings,mark=at position .55 with {\arrow{Stealth[length=2.2mm]}}}},
 rlip/.style={red!78!black,line width=.82pt,postaction={decorate},decoration={markings,mark=at position .55 with {\arrow{Stealth[length=1.7mm]}}}},
 rholip/.style={blue!78!black,line width=.82pt,postaction={decorate},decoration={markings,mark=at position .55 with {\arrow{Stealth[length=1.7mm]}}}},
 guide/.style={gray!65,dashed,line width=.45pt}]

% Unit circle and real axis.
\draw[guide] (0,0) circle (1);
\draw[gray!55,line width=.35pt] (-1.18,0)--(1.18,0);
\fill (0,0) circle (.8pt); \node[below right,scale=.55] at (0,0) {$0$};
\node[right,scale=.55] at (1,0) {$1$};
\node[left,scale=.55] at (-1,0) {$-1$};

% Representative upper endpoint angles: N=5, active ell=3.
\def\aone{13}\def\atwo{25}
\def\athree{38}\def\afour{51}
\def\afive{65}\def\aalpha{73}\def\asix{82}
\def\aseven{98}\def\aeight{111}
\def\anine{127}\def\aten{143}
\def\rout{1.105}\def\rin{0.895}

% Helper: bands.
\foreach \aa/\bb in {\aone/\atwo,\athree/\afour,\afive/\asix,\aseven/\aeight,\anine/\aten}{
  \draw[band] (\aa:1) arc[start angle=\aa,end angle=\bb,radius=1];
  \draw[band] (-\bb:1) arc[start angle=-\bb,end angle=-\aa,radius=1];
}

% r-lenses on j=1,2 (upper and reflected lower).
\foreach \aa/\bb in {\aone/\atwo,\athree/\afour}{
  \draw[rlip] (\aa:\rout) arc[start angle=\aa,end angle=\bb,radius=\rout];
  \draw[rlip] (\aa:\rin) arc[start angle=\aa,end angle=\bb,radius=\rin];
  \draw[red!78!black,line width=.65pt] (\aa:\rin)--(\aa:\rout) (\bb:\rin)--(\bb:\rout);
  \draw[rlip] (-\bb:\rout) arc[start angle=-\bb,end angle=-\aa,radius=\rout];
  \draw[rlip] (-\bb:\rin) arc[start angle=-\bb,end angle=-\aa,radius=\rin];
  \draw[red!78!black,line width=.65pt] (-\aa:\rin)--(-\aa:\rout) (-\bb:\rin)--(-\bb:\rout);
}

% Active r-part (eta_5 to alpha-epsilon) and rho-part (alpha+epsilon to eta_6).
\def\am{71.7}\def\ap{74.3}
\draw[rlip] (\afive:\rout) arc[start angle=\afive,end angle=\am,radius=\rout];
\draw[rlip] (\afive:\rin) arc[start angle=\afive,end angle=\am,radius=\rin];
\draw[red!78!black,line width=.65pt] (\afive:\rin)--(\afive:\rout);
\draw[rholip] (\ap:\rout) arc[start angle=\ap,end angle=\asix,radius=\rout];
\draw[rholip] (\ap:\rin) arc[start angle=\ap,end angle=\asix,radius=\rin];
\draw[blue!78!black,line width=.65pt] (\asix:\rin)--(\asix:\rout);
% Crossed joining at alpha.
\draw[red!78!black,line width=.8pt] (\am:\rout)--(\ap:\rin);
\draw[blue!78!black,line width=.8pt] (\am:\rin)--(\ap:\rout);
% Lower reflected active opening.
\draw[rlip] (-\am:\rout) arc[start angle=-\am,end angle=-\afive,radius=\rout];
\draw[rlip] (-\am:\rin) arc[start angle=-\am,end angle=-\afive,radius=\rin];
\draw[red!78!black,line width=.65pt] (-\afive:\rin)--(-\afive:\rout);
\draw[rholip] (-\asix:\rout) arc[start angle=-\asix,end angle=-\ap,radius=\rout];
\draw[rholip] (-\asix:\rin) arc[start angle=-\asix,end angle=-\ap,radius=\rin];
\draw[blue!78!black,line width=.65pt] (-\asix:\rin)--(-\asix:\rout);
\draw[red!78!black,line width=.8pt] (-\am:\rout)--(-\ap:\rin);
\draw[blue!78!black,line width=.8pt] (-\am:\rin)--(-\ap:\rout);

% rho-lenses on j=4,5, including all later arcs.
\foreach \aa/\bb in {\aseven/\aeight,\anine/\aten}{
  \draw[rholip] (\aa:\rout) arc[start angle=\aa,end angle=\bb,radius=\rout];
  \draw[rholip] (\aa:\rin) arc[start angle=\aa,end angle=\bb,radius=\rin];
  \draw[blue!78!black,line width=.65pt] (\aa:\rin)--(\aa:\rout) (\bb:\rin)--(\bb:\rout);
  \draw[rholip] (-\bb:\rout) arc[start angle=-\bb,end angle=-\aa,radius=\rout];
  \draw[rholip] (-\bb:\rin) arc[start angle=-\bb,end angle=-\aa,radius=\rin];
  \draw[blue!78!black,line width=.65pt] (-\aa:\rin)--(-\aa:\rout) (-\bb:\rin)--(-\bb:\rout);
}

% Endpoints and labels.
\foreach \ang in {\aone,\atwo,\athree,\afour,\afive,\asix,\aseven,\aeight,\anine,\aten}{
  \fill (\ang:1) circle (.72pt); \fill (-\ang:1) circle (.72pt);
}
\fill (\aalpha:1) circle (.92pt); \fill (-\aalpha:1) circle (.92pt);
\node[scale=.53,anchor=south east] at (\aalpha:1.03) {$\alpha_\ell$};
\node[scale=.53,anchor=north east] at (-\aalpha:1.03) {$\alpha_\ell^{-1}$};
\node[scale=.55] at (19:1.25) {$\Gamma_1$};
\node[scale=.55] at (44.5:1.25) {$\Gamma_{\ell-1}$};
\node[scale=.55] at (73.5:1.30) {$\Gamma_\ell$};
\node[scale=.55] at (104.5:1.25) {$\Gamma_{\ell+1}$};
\node[scale=.55] at (135:1.25) {$\Gamma_N$};
\node[scale=.50,red!70!black] at (50:.58) {$r$-lenses, $j<\ell$};
\node[scale=.50,blue!70!black] at (120:.59) {$\rho$-lenses, $j>\ell$};
\node[scale=.48,align=center] at (73:.47) {crossed\\opening};
\node[scale=.48,red!70!black] at (67:1.43) {$r$};
\node[scale=.48,blue!70!black] at (80:1.43) {$\rho$};
\end{tikzpicture}
\caption{All-band lens opening in the mixed sector $\mathcal M_\ell$ (representative case $N=5$, $\ell=3$).  The active band $\Gamma_\ell$ is split at $\alpha_\ell$ and the $r$- and $\rho$-lenses are joined crosswise.  Crucially, full $\rho$-lenses are also opened around every remaining band $\Gamma_j$, $j>\ell$.  The lower configuration is fixed by the dNLS symmetry.}
\label{fig:all-band-lenses}
\end{figure}

\begin{proposition}[Full-contour reduction]\label{prop:full-reduction}
Let $T$ be defined by \eqref{eq:T-transform}.  Uniformly for $\xi$ in a compact subset of a pure or mixed sector, its jumps satisfy
\begin{equation}\label{eq:T-jumps}
 J_T(z)=
 \begin{cases}
 -\ii\sigma_1,&z\in\displaystyle\bigcup_{j=1}^N\Gamma_j,\\[1mm]
 \ii\sigma_1,&z\in\displaystyle\bigcup_{j=1}^N\ol\Gamma_j,\\[1mm]
 \ee^{\ii\vartheta_{N,m}\sigma_3},&z\in\mathfrak g_m\cup\ol{\mathfrak g}_m,
 \quad m=1,\ldots,N,\\[1mm]
 I+O(\ee^{-ct}),&z\text{ on all lens lips outside local discs},
 \end{cases}
\end{equation}
where
\begin{equation}\label{eq:vartheta}
 \vartheta_{N,m}=t\Omega_{N,m}(\xi)+\Delta_{N,m}(\xi).
\end{equation}
In a mixed sector, the two central factorizations on the portions of $\Gamma_\ell$ have the same middle factor $-\ii\sigma_1$.  Hence the switching point does not break the central model contour: the jump $-\ii\sigma_1$ extends across the whole of $\Gamma_\ell$.  Therefore the nondecaying model contour contains all $\Gamma_j$, $j=1,\ldots,N$.
\end{proposition}

\begin{proof}
Equations \eqref{eq:S-r-factor} and \eqref{eq:S-rho-factor} show that the central factor is independent of which triangular factorization is used.  The transformation \eqref{eq:T-transform} moves the two outer triangular factors to the lens lips.  Proposition~\ref{prop:signs} makes these lip jumps exponentially close to the identity.  The gap jumps come from \eqref{eq:gN-gap} and \eqref{eq:delta-gap}.  The same argument applies to every band, including all $j>\ell$; no spectral arc is discarded.
\end{proof}

\section{The fixed genus-\texorpdfstring{$N$}{N} model problem}\label{sec:model}

\subsection{Model problem in the \texorpdfstring{$z$}{z}-plane}

\begin{RHP}[Fixed full model in the $z$-plane]\label{rhp:zmodel}
Find a $2\times2$ matrix $Y_N(z)=Y_N(z;x,t)$ such that:
\begin{enumerate}[label=(\roman*)]
\item $Y_N$ is analytic in the complement of all upper and lower bands, all upper and lower gaps, and $z=0$.
\item
\begin{equation}\label{eq:Y-normalization}
 Y_N(z)=I+O(z^{-1}),\quad z\to\infty,
 \qquad
 Y_N(z)=\frac{\sigma_1}{z}+O(1),\quad z\to0.
\end{equation}
\item
\begin{equation}\label{eq:Y-symmetry}
 Y_N(z)=\sigma_1\ol{Y_N(\ol z)}\sigma_1
 =z^{-1}Y_N(z^{-1})\sigma_1.
\end{equation}
\item Its jumps are
\begin{equation}\label{eq:Y-jumps}
 Y_{N,+}=Y_{N,-}
 \begin{cases}
 -\ii\sigma_1,&z\in\displaystyle\bigcup_{j=1}^N\Gamma_j,\\
 \ii\sigma_1,&z\in\displaystyle\bigcup_{j=1}^N\ol\Gamma_j,\\
 \ee^{\ii\vartheta_{N,m}\sigma_3},&z\in\mathfrak g_m\cup\ol{\mathfrak g}_m,
 \quad m=1,\ldots,N.
 \end{cases}
\end{equation}
\end{enumerate}
\end{RHP}

The RHP is the same in all $2N+1$ sectors.  Only $\Delta_N$, and therefore the vector $\vartheta_N$, changes with the factorization data.

\subsection{Joukowski folding and completion of all model curves}

Let
\begin{equation}\label{eq:z-k}
 z=z(k)=k+\sqrt{k^2-1},\qquad z(k)\sim2k\quad(k\to\infty),
\end{equation}
and set
\begin{equation}\label{eq:F-def}
 F_N(k)=\frac{1}{\sqrt{1-z(k)^{-2}}}\,Y_N(z(k)).
\end{equation}
The $N$ unit-circle spectral bands become gaps in the $k$-plane, while their $N+1$ complementary arcs become the $N+1$ cuts
\begin{align}
 \Sigma_0&=(e_1,1),\label{eq:Sigma0}\\
 \Sigma_m&=(e_{2m+1},e_{2m}),\qquad m=1,\ldots,N-1,\label{eq:Sigmam}\\
 \Sigma_N&=(-1,e_{2N}).\label{eq:SigmaN}
\end{align}
Figure~\ref{fig:full-joukowski} displays both the complete post-lens contour in the $z$-plane and all $N+1$ model cuts in the $k$-plane (the drawing uses $N=4$ only to keep the labels legible).

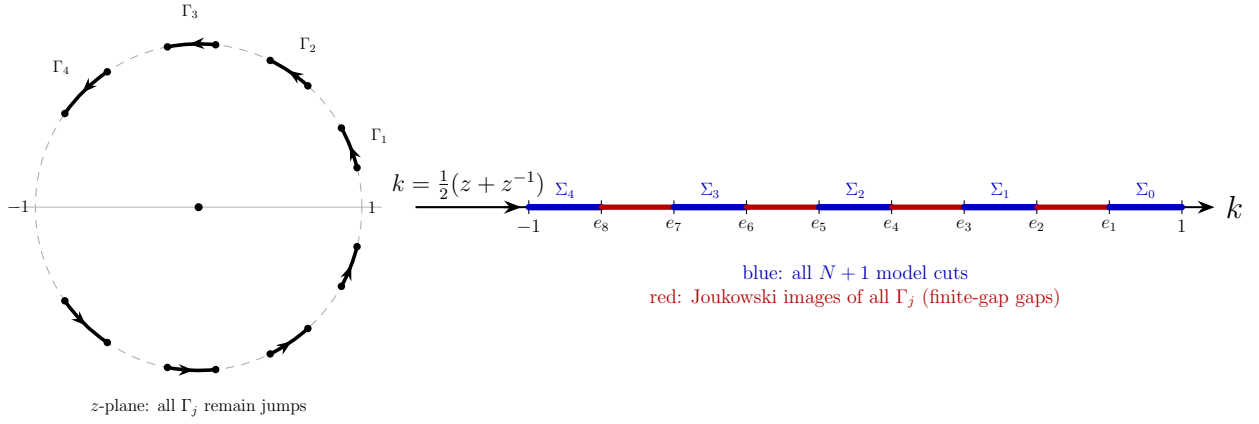
\begin{figure}[htbp]
\centering
\begin{tikzpicture}[scale=.96,>=Stealth,line cap=round,line join=round,
 zband/.style={black,line width=1.35pt,postaction={decorate},decoration={markings,mark=at position .55 with {\arrow{Stealth[length=2mm]}}}},
 kgap/.style={red!70!black,line width=2.0pt},
 kband/.style={blue!72!black,line width=2.4pt},
 mapline/.style={gray!55,dashed,line width=.55pt}]
% z-plane panel (representative N=4; all four bands retained).
\begin{scope}[xshift=-4.9cm,yshift=.15cm,scale=2.25]
 \draw[gray!65,dashed] (0,0) circle (1);
 \draw[gray!55] (-1.12,0)--(1.12,0);
 \fill (0,0) circle (.7pt);
 \node[scale=.55,below] at (0,-1.14) {$z$-plane: all $\Gamma_j$ remain jumps};
 \node[scale=.55,right] at (1,0) {$1$};
 \node[scale=.55,left] at (-1,0) {$-1$};
 \def\bA{14}\def\bB{29}\def\bC{48}\def\bD{64}\def\bE{84}\def\bF{101}\def\bG{124}\def\bH{145}
 \foreach \aa/\bb in {\bA/\bB,\bC/\bD,\bE/\bF,\bG/\bH}{
  \draw[zband] (\aa:1) arc[start angle=\aa,end angle=\bb,radius=1];
  \draw[zband] (-\bb:1) arc[start angle=-\bb,end angle=-\aa,radius=1];
  \fill (\aa:1) circle (.65pt); \fill (\bb:1) circle (.65pt);
  \fill (-\aa:1) circle (.65pt); \fill (-\bb:1) circle (.65pt);
 }
 \node[scale=.52] at (21.5:1.19) {$\Gamma_1$};
 \node[scale=.52] at (56:1.20) {$\Gamma_2$};
 \node[scale=.52] at (92.5:1.20) {$\Gamma_3$};
 \node[scale=.52] at (134.5:1.20) {$\Gamma_4$};
\end{scope}

% Arrow between panels.
\draw[->,line width=.8pt] (-1.9,.15)--(-.45,.15);
\node[above,scale=.76] at (-1.18,.15) {$k=\frac12(z+z^{-1})$};

% k-plane panel, representative N=4.  Blue intervals are the N+1 model cuts.
\begin{scope}[xshift=4.15cm,yshift=.15cm]
 \draw[->,line width=.75pt] (-4.85,0)--(4.95,0) node[right] {$k$};
 \foreach \x in {-4.5,-3.5,-2.5,-1.5,-.5,.5,1.5,2.5,3.5,4.5}{\draw (\x,.09)--(\x,-.09);}
 \node[below,scale=.59] at (-4.5,-.09) {$-1$};
 \node[below,scale=.59] at (-3.5,-.09) {$e_8$};
 \node[below,scale=.59] at (-2.5,-.09) {$e_7$};
 \node[below,scale=.59] at (-1.5,-.09) {$e_6$};
 \node[below,scale=.59] at (-.5,-.09) {$e_5$};
 \node[below,scale=.59] at (.5,-.09) {$e_4$};
 \node[below,scale=.59] at (1.5,-.09) {$e_3$};
 \node[below,scale=.59] at (2.5,-.09) {$e_2$};
 \node[below,scale=.59] at (3.5,-.09) {$e_1$};
 \node[below,scale=.59] at (4.5,-.09) {$1$};
 % Correct alternating band-gap conversion, from left to right.
 \draw[kband] (-4.5,0)--(-3.5,0);  % Sigma_4
 \draw[kgap]  (-3.5,0)--(-2.5,0);  % image Gamma_4
 \draw[kband] (-2.5,0)--(-1.5,0);  % Sigma_3
 \draw[kgap]  (-1.5,0)--(-.5,0);   % image Gamma_3
 \draw[kband] (-.5,0)--(.5,0);     % Sigma_2
 \draw[kgap]  (.5,0)--(1.5,0);     % image Gamma_2
 \draw[kband] (1.5,0)--(2.5,0);    % Sigma_1
 \draw[kgap]  (2.5,0)--(3.5,0);    % image Gamma_1
 \draw[kband] (3.5,0)--(4.5,0);    % Sigma_0
 \node[above,scale=.56,blue!72!black] at (-4.0,.04) {$\Sigma_4$};
 \node[above,scale=.56,blue!72!black] at (-2.0,.04) {$\Sigma_3$};
 \node[above,scale=.56,blue!72!black] at (0,.04) {$\Sigma_2$};
 \node[above,scale=.56,blue!72!black] at (2.0,.04) {$\Sigma_1$};
 \node[above,scale=.56,blue!72!black] at (4.0,.04) {$\Sigma_0$};
 \node[scale=.62,blue!72!black] at (0,-.9) {blue: all $N+1$ model cuts};
 \node[scale=.62,red!70!black] at (0,-1.25) {red: Joukowski images of all $\Gamma_j$ (finite-gap gaps)};
\end{scope}
\end{tikzpicture}
\caption{Complete contour reduction.  In the $z$-plane the reduced model has the constant central jump on every $\Gamma_j$, including all $j>\ell$.  Under the Joukowski map, these $N$ arcs become finite-gap gaps (red), while the complementary arcs become the $N+1$ real cuts $\Sigma_0,\ldots,\Sigma_N$ (blue) of the fixed genus-$N$ model.}
\label{fig:full-joukowski}
\end{figure}

\begin{RHP}[Fixed full model in the $k$-plane]\label{rhp:kmodel}
Find $F_N(k)$ analytic in $\C\setminus\bigcup_{m=0}^{N}\Sigma_m$, normalized by $F_N(k)=I+O(k^{-1})$ at infinity, and satisfying
\begin{equation}\label{eq:F-jumps}
 F_{N,+}(k)=F_{N,-}(k)
 \begin{cases}
 \begin{pmatrix}0&-\ii\\-\ii&0\end{pmatrix},&k\in\Sigma_0,\\[4mm]
 \begin{pmatrix}0&-\ii\ee^{\ii\vartheta_{N,m}}\\
 -\ii\ee^{-\ii\vartheta_{N,m}}&0\end{pmatrix},&k\in\Sigma_m,
 \quad m=1,\ldots,N.
 \end{cases}
\end{equation}
It also satisfies $F_N(k)=\sigma_1\ol{F_N(\ol k)}\sigma_1$.
\end{RHP}

\subsection{Canonical \texorpdfstring{$k$}{k}-plane cycles and the Baker--Akhiezer construction}\label{subsec:BA-model}

The notation in this subsection is intrinsic to the quotient $k$-surface and is kept distinct from the cycles used in the $z$-plane $g$-function construction.  Put
\begin{equation}\label{eq:k-branch-order}
 e_0=1>e_1>\cdots>e_{2N}>e_{2N+1}=-1
\end{equation}
and compactify the affine curve
\begin{equation}\label{eq:XN-curve}
 \mathscr X_N:\qquad
 w^2=\prod_{\nu=0}^{2N+1}(k-e_\nu)=\mathscr R_N(k)^2
\end{equation}
by the two points $\infty_\pm$.  The first sheet $\mathscr X_N^+$ is fixed by
$w=\mathscr R_N(k)\sim k^{N+1}$ at $\infty_+$, and
$P^*=(k,-w)$ denotes the hyperelliptic involution.  We realize
$\mathscr X_N$ by gluing two copies of the $k$-plane along
\begin{equation}\label{eq:k-cuts-all}
 \Sigma_m=[e_{2m+1},e_{2m}],\qquad m=0,\ldots,N.
\end{equation}
Thus $\Sigma_0=[e_1,e_0]$ is the reference cut and the remaining
$N$ cuts are $\Sigma_1,\ldots,\Sigma_N$.

For $j=1,\ldots,N$, let $\widetilde a_j$ be the positive oval around
$\Sigma_j$.  Let $\widetilde b_j$ start on the upper bank of
$\Sigma_0$, run on the first sheet to the upper bank of $\Sigma_j$, cross
$\Sigma_j$ to the second sheet, return to $\Sigma_0$ on the second sheet,
and cross $\Sigma_0$ back to the first sheet.  The orientations are chosen so
that
\begin{equation}\label{eq:k-cycle-intersection}
 \widetilde a_j\circ\widetilde a_m=0,
 \qquad
 \widetilde b_j\circ\widetilde b_m=0,
 \qquad
 \widetilde a_j\circ\widetilde b_m=\delta_{jm}.
\end{equation}
Figure~\ref{fig:k-canonical-basis} gives the projection of this basis to the
$k$-plane.  Solid red pieces are traversed on the first sheet and dashed red
pieces on the second sheet.

\begin{figure}[htbp]
\centering
\begin{tikzpicture}[x=1.08cm,y=.98cm,>=Stealth,line cap=round,line join=round,
 cut/.style={blue!72!black,line width=2.8pt},
 acycle/.style={blue!72!black,line width=1.05pt},
 bcycle/.style={red!78!black,line width=1.05pt},
 bback/.style={red!78!black,dashed,line width=1.0pt},
 bridge/.style={red!78!black,line width=.82pt}]
 % A representative genus-N configuration, with only the extreme cycles shown.
 \draw[->,line width=.72pt] (-5.15,0)--(5.18,0) node[right] {$k$};
 \foreach \x/\lab in {-4.7/{e_{2N+1}=-1},-3.8/{e_{2N}},-2.8/{e_{2N-1}},-1.9/{e_{2N-2}},-0.9/{\cdots},0/{\cdots},1/{e_3},1.9/{e_2},2.9/{e_1},3.8/{e_0=1}}{
   \fill (\x,0) circle (1.08pt);
   \node[below=2.5pt,scale=.62] at (\x,0) {$\lab$};
 }
 \draw[cut] (-4.7,0)--(-3.8,0) node[midway,above=3pt,scale=.66] {$\Sigma_N$};
 \draw[cut] (-2.8,0)--(-1.9,0) node[midway,above=3pt,scale=.66] {$\Sigma_{N-1}$};
 \draw[cut] (-0.9,0)--(0,0) node[midway,above=3pt,scale=.66] {$\cdots$};
 \draw[cut] (1,0)--(1.9,0) node[midway,above=3pt,scale=.66] {$\Sigma_1$};
 \draw[cut] (2.9,0)--(3.8,0) node[midway,above=3pt,scale=.66] {$\Sigma_0$};

 % Positive a-ovals.  The omitted ovals are arranged in the same way.
 \draw[acycle,postaction={decorate},decoration={markings,mark=at position .18 with {\arrow{Stealth[length=1.9mm]}}}]
   (-4.25,0) ellipse (.66 and .48);
 \node[blue!72!black,scale=.72] at (-4.25,.70) {$\widetilde a_N$};
 \draw[acycle,postaction={decorate},decoration={markings,mark=at position .18 with {\arrow{Stealth[length=1.9mm]}}}]
   (1.45,0) ellipse (.66 and .48);
 \node[blue!72!black,scale=.72] at (1.45,.70) {$\widetilde a_1$};

 % Inner b_1 cycle: the upper part is on X_N^+, the lower return on X_N^-.
 \draw[bcycle,postaction={decorate},decoration={markings,mark=at position .58 with {\arrow{Stealth[length=1.9mm]}}}]
   (3.18,.11) .. controls (2.92,.72) and (2.06,.76) .. (1.45,.11);
 \draw[bridge] (1.45,.11)--(1.45,-.11);
 \draw[bback,postaction={decorate},decoration={markings,mark=at position .42 with {\arrow{Stealth[length=1.9mm]}}}]
   (1.45,-.11) .. controls (2.06,-.64) and (2.92,-.62) .. (3.18,-.11);
 \draw[bridge] (3.18,-.11)--(3.18,.11);
 \node[red!78!black,scale=.72] at (2.32,.84) {$\widetilde b_1$};

 % Outer b_N cycle.  Its projection is nested strictly outside b_1, so the b-cycles do not meet.
 \draw[bcycle,postaction={decorate},decoration={markings,mark=at position .60 with {\arrow{Stealth[length=1.9mm]}}}]
   (3.67,.14) .. controls (3.48,1.62) and (-2.55,1.76) .. (-4.25,.14);
 \draw[bridge] (-4.25,.14)--(-4.25,-.14);
 \draw[bback,postaction={decorate},decoration={markings,mark=at position .40 with {\arrow{Stealth[length=1.9mm]}}}]
   (-4.25,-.14) .. controls (-2.55,-1.48) and (3.48,-1.36) .. (3.67,-.14);
 \draw[bridge] (3.67,-.14)--(3.67,.14);
 \node[red!78!black,scale=.72] at (-.60,1.73) {$\widetilde b_N$};

 % Legend and sheet labels.
 \node[scale=.61,anchor=west] at (-4.9,-1.80)
 {\textcolor{red!78!black}{solid}: first sheet $\mathscr X_N^+$};
 \node[scale=.61,anchor=west] at (-1.35,-1.80)
 {\textcolor{red!78!black}{dashed}: second sheet $\mathscr X_N^-$};
 \node[scale=.58,blue!72!black] at (3.42,.92) {reference cut};
\end{tikzpicture}
\caption{A nonintersecting projection of the canonical homology basis
$\{\widetilde a_j,\widetilde b_j\}_{j=1}^N$ on the fixed hyperelliptic
$k$-surface.  Only the extreme pairs $(\widetilde a_1,\widetilde b_1)$ and
$(\widetilde a_N,\widetilde b_N)$ are displayed.  The $b$-cycles are drawn as
nested paths with distinct crossing points on the reference cut $\Sigma_0$;
hence their projections do not intersect one another.  Solid red portions
lie on $\mathscr X_N^+$ and dashed red portions lie on $\mathscr X_N^-$.}
\label{fig:k-canonical-basis}
\end{figure}
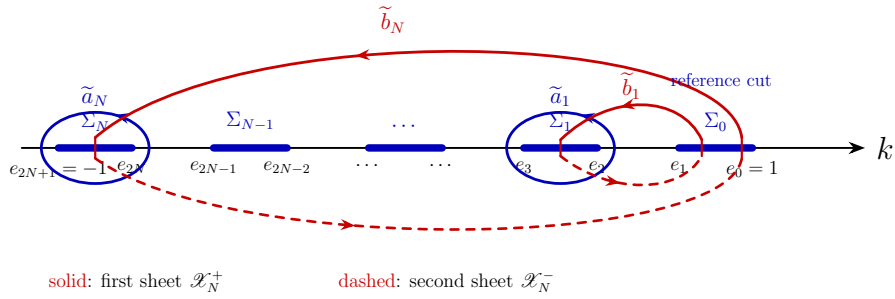

Let $\dd\widetilde\omega_m$, $m=1,\ldots,N$, be the normalized holomorphic
differentials
\begin{equation}\label{eq:k-holo-differentials}
 \dd\widetilde\omega_m(P)
 =\frac{\sum_{q=0}^{N-1}c_{mq}k^q}{w(P)}\,\dd k,
 \qquad
 \oint_{\widetilde a_j}\dd\widetilde\omega_m=\delta_{jm}.
\end{equation}
The $k$-plane period matrix and Abel map are denoted by new symbols,
\begin{equation}\label{eq:k-period-matrix}
 (\widetilde\tau_N)_{jm}
 :=\oint_{\widetilde b_j}\dd\widetilde\omega_m,
 \qquad
 \widetilde{\mathbf J}_N(P)
 :=\int_{P_0}^{P}
 (\dd\widetilde\omega_1,\ldots,\dd\widetilde\omega_N)^T,
 \qquad P_0=(e_0,0).
\end{equation}
The path in \eqref{eq:k-period-matrix} is taken in the canonical polygon
obtained by cutting $\mathscr X_N$ along the cycles in
Figure~\ref{fig:k-canonical-basis}.  Then $\widetilde\tau_N$ is symmetric and
$\Ima\widetilde\tau_N>0$.  For this real $M$-curve and the displayed real canonical basis one may choose the normalized differentials so that $\widetilde\tau_N=\ii Y_N$ with $Y_N$ real symmetric positive definite.  With the above path convention, the boundary values
of the Abel map on the first sheet obey
\begin{align}
 \widetilde{\mathbf J}_{N,+}(k)+
 \widetilde{\mathbf J}_{N,-}(k)&=0,
 &&k\in\Sigma_0,\label{eq:Abel-jump-Sigma0}\\
 \widetilde{\mathbf J}_{N,+}(k)+
 \widetilde{\mathbf J}_{N,-}(k)&=-\widetilde{\boldsymbol\tau}_{N,j},
 &&k\in\Sigma_j,\quad j=1,\ldots,N,\label{eq:Abel-jump-Sigmaj}
\end{align}
modulo $\mathbb Z^N$, where $\widetilde{\boldsymbol\tau}_{N,j}$ denotes the
$j$th column of $\widetilde\tau_N$.  To agree with the notation used in the
statement of the main theorem, from now on we set
\begin{equation}\label{eq:identify-k-notation}
 \tauN:=\widetilde\tau_N,
 \qquad
 \JN(P):=\widetilde{\mathbf J}_N(P),
 \qquad
 \mathbf J_\infty:=\JN(\infty_+).
\end{equation}

\paragraph{The Baker--Akhiezer divisor.}
Define the two monic polynomials
\begin{equation}\label{eq:Pi-even-odd}
 \Pi_{\rm e}(k):=\prod_{m=0}^{N}(k-e_{2m}),
 \qquad
 \Pi_{\rm o}(k):=\prod_{m=0}^{N}(k-e_{2m+1}),
\end{equation}
and the fourth-root function
\begin{equation}\label{eq:gammaN-BA}
 \gamma_N(k):=
 \left(\frac{\Pi_{\rm e}(k)}{\Pi_{\rm o}(k)}\right)^{1/4},
 \qquad
 \gamma_N(k)=1+O(k^{-1}),\quad k\to\infty.
\end{equation}
Its branch cuts are precisely $\Sigma_0,\ldots,\Sigma_N$, and, with every cut
oriented from left to right,
\begin{equation}\label{eq:gamma-boundary}
 \gamma_{N,+}(k)=\ii\gamma_{N,-}(k),
 \qquad k\in\bigcup_{m=0}^{N}\Sigma_m.
\end{equation}
For $j=1,\ldots,N$, put
\begin{equation}\label{eq:k-spectral-gaps}
 \mathcal G_j=(e_{2j},e_{2j-1}).
\end{equation}
The symbols used in the velocity proof for the zeros of $\mathsf P_N$ are
\emph{not} used here.  Instead, let $\pi_j^{\rm BA}$ be the unique zero in
$\mathcal G_j$ of
\begin{equation}\label{eq:BA-root-equation}
 \Pi_{\rm e}(k)-\Pi_{\rm o}(k)=0,
\end{equation}
and let
\begin{equation}\label{eq:BA-points}
 P_j^{\rm BA}:=
 \bigl(\pi_j^{\rm BA},\mathscr R_N(\pi_j^{\rm BA})\bigr)
 \in\mathscr X_N^+.
\end{equation}
Thus $P_j^{\rm BA}$, rather than the zeros of the second-kind polynomial
$\mathsf P_N$, are the divisor points entering the Baker--Akhiezer matrix.

\begin{lemma}[Location and meaning of the Baker--Akhiezer points]
\label{lem:BA-points}
Equation \eqref{eq:BA-root-equation} has exactly one simple root
$\pi_j^{\rm BA}$ in each gap $\mathcal G_j$.  On the first sheet these points
are precisely the finite zeros of $\gamma_N-\gamma_N^{-1}$.  More
intrinsically, the meromorphic function
\begin{equation}\label{eq:chiN-meromorphic}
 \chi_N(P):=
 \frac{\gamma_N(P)-\gamma_N(P)^{-1}}
      {\gamma_N(P)+\gamma_N(P)^{-1}}
 =\frac{w(P)-\Pi_{\rm o}(k)}{w(P)+\Pi_{\rm o}(k)}
\end{equation}
has divisor
\begin{equation}\label{eq:chiN-divisor}
 (\chi_N)=
 \infty_+ +\sum_{j=1}^{N}P_j^{\rm BA}
 -\infty_- -\sum_{j=1}^{N}(P_j^{\rm BA})^*.
\end{equation}
\end{lemma}

\begin{proof}
On $\mathcal G_j$ the quotient $\Pi_{\rm e}/\Pi_{\rm o}$ is positive and
continuous, tends to $0$ at $e_{2j}$, and tends to $+\infty$ at
$e_{2j-1}$.  Hence \eqref{eq:BA-root-equation} has at least one root in every
$\mathcal G_j$.  Since the leading terms of the two monic degree-$N+1$
polynomials cancel, $\deg(\Pi_{\rm e}-\Pi_{\rm o})\le N$.  The $N$ disjoint
gaps therefore contain all its zeros, one in each gap, and every zero is
simple.  On $\mathscr X_N^+$ one has
$\gamma_N^2=w/\Pi_{\rm o}$; this gives \eqref{eq:chiN-meromorphic} and shows
that the finite zeros are exactly the points \eqref{eq:BA-points}.  The
expansions at $\infty_\pm$ and the hyperelliptic involution then give
\eqref{eq:chiN-divisor}.
\end{proof}

Let $\widetilde{\mathbf K}_N$ be the vector of Riemann constants associated
with the base point $P_0$ and the basis
$\{\widetilde a_j,\widetilde b_j\}$.  Define
\begin{equation}\label{eq:dN-rigorous}
 \mathbf d_N
 :=\sum_{j=1}^{N}\JN(P_j^{\rm BA})-\widetilde{\mathbf K}_N.
\end{equation}
We use the sign convention for the vector of Riemann constants for which Riemann's vanishing theorem identifies the zero divisor of
$\Theta(\JN(P)-\mathbf d_N)$.  Abel's theorem applied to \eqref{eq:chiN-divisor}, together with Riemann's vanishing theorem, yields the identity
\begin{equation}\label{eq:dN-Abel-identity}
 \mathbf d_N\equiv-\JN(\infty_+)
 \pmod{\mathbb Z^N+\tauN\mathbb Z^N}.
\end{equation}
In particular, the theta denominators below have the Baker--Akhiezer divisor
$P_1^{\rm BA}+\cdots+P_N^{\rm BA}$, and the zeros of
$\gamma_N-\gamma_N^{-1}$ cancel their apparent poles.  We choose representatives in the Jacobian so that \eqref{eq:dN-Abel-identity} is exact; period-lattice multipliers are then absorbed into the diagonal normalization.  Since $\tauN=\ii Y_N$ with $Y_N>0$, Poisson summation gives
\begin{equation}\label{eq:theta-positive-real}
 \ThetaN(\mathbf x;\ii Y_N)
 =\frac{1}{\sqrt{\det Y_N}}
 \sum_{\mathbf m\in\mathbb Z^N}
 \exp\!\left[-\pi(\mathbf m+\mathbf x)^TY_N^{-1}(\mathbf m+\mathbf x)\right]>0,
 \qquad \mathbf x\in\mathbb R^N.
\end{equation}
Thus the normalizing theta factors below do not meet the theta divisor for the real phase vectors arising in the theorem.  Standard background on these divisor identities may be found in \cite{FarkasKra1992,Fay1973}.

\paragraph{Explicit Baker--Akhiezer matrix.}
Define
\begin{equation}\label{eq:thetaN-rigorous}
 \ThetaN(\mathbf z;\tauN)=
 \sum_{\mathbf n\in\mathbb Z^N}
 \exp\!\left(
 \pi\ii\mathbf n^T\tauN\mathbf n+2\pi\ii\mathbf n^T\mathbf z
 \right),
 \qquad \Theta(\mathbf z):=\ThetaN(\mathbf z;\tauN),
\end{equation}
and set
\begin{equation}\label{eq:ZN-rigorous}
 \mathbf Z_N
 :=\frac{1}{2\pi}
 \bigl(\vartheta_{N,1},\ldots,\vartheta_{N,N}\bigr)^T
 =\frac{t\Omega_N(\xi)+\Delta_N(\xi)}{2\pi}.
\end{equation}
Introduce the diagonal normalization
\begin{equation}\label{eq:normalizerN-rigorous}
 \mathcal N_N(\mathbf Z_N)=
 \begin{pmatrix}
 \dfrac{\Theta(\mathbf J_\infty+\mathbf d_N)}
 {\Theta(-\mathbf Z_N+\mathbf J_\infty+\mathbf d_N)}&0\\[4mm]
 0&\dfrac{\Theta(-\mathbf J_\infty-\mathbf d_N)}
 {\Theta(-\mathbf Z_N-\mathbf J_\infty-\mathbf d_N)}
 \end{pmatrix}.
\end{equation}
For $P=k^+\in\mathscr X_N^+$, write $\JN(k):=\JN(P)$ and define
\begin{align}
 \mathcal B_{11}(k)&=
 \frac{\gamma_N(k)+\gamma_N(k)^{-1}}2
 \frac{\Theta(-\mathbf Z_N+\JN(k)+\mathbf d_N)}
      {\Theta(\JN(k)+\mathbf d_N)},\label{eq:B11N-rigorous}\\
 \mathcal B_{12}(k)&=-
 \frac{\gamma_N(k)-\gamma_N(k)^{-1}}2
 \frac{\Theta(-\mathbf Z_N-\JN(k)+\mathbf d_N)}
      {\Theta(-\JN(k)+\mathbf d_N)},\label{eq:B12N-rigorous}\\
 \mathcal B_{21}(k)&=-
 \frac{\gamma_N(k)-\gamma_N(k)^{-1}}2
 \frac{\Theta(-\mathbf Z_N+\JN(k)-\mathbf d_N)}
      {\Theta(\JN(k)-\mathbf d_N)},\label{eq:B21N-rigorous}\\
 \mathcal B_{22}(k)&=
 \frac{\gamma_N(k)+\gamma_N(k)^{-1}}2
 \frac{\Theta(-\mathbf Z_N-\JN(k)-\mathbf d_N)}
      {\Theta(-\JN(k)-\mathbf d_N)}.\label{eq:B22N-rigorous}
\end{align}
The model solution is
\begin{equation}\label{eq:FN-solution-rigorous}
 F_N(k)=\mathcal N_N(\mathbf Z_N)
 \begin{pmatrix}
  \mathcal B_{11}(k)&\mathcal B_{12}(k)\\
  \mathcal B_{21}(k)&\mathcal B_{22}(k)
 \end{pmatrix}.
\end{equation}

\begin{proposition}[Solution and expansion of the fixed model]
\label{prop:model}
The matrix \eqref{eq:FN-solution-rigorous} is the unique solution of
RHP~\ref{rhp:kmodel}.  Moreover, the inverse Joukowski transformation is the piecewise formula
\begin{equation}\label{eq:YN-from-FN}
 Y_N(z)=
 \begin{cases}
 \sqrt{1-z^{-2}}\,
 F_N\!\left(\dfrac{z+z^{-1}}2\right),&|z|>1,\\[3mm]
 z^{-1}\sqrt{1-z^2}\,
 F_N\!\left(\dfrac{z+z^{-1}}2\right)\sigma_1,&|z|<1,
 \end{cases}
\end{equation}
with compatible branches on the two domains.  It solves RHP~\ref{rhp:zmodel}.  As $z\to\infty$,
\begin{equation}\label{eq:YN-expansion}
 Y_N(z)=I+\frac{Y_{N,1}}{z}+O(z^{-2}),
\end{equation}
and
\begin{equation}\label{eq:YN21}
 (Y_{N,1})_{21}=\AN\TN(\mathbf Z_N),
\end{equation}
where $\AN$ and $\TN$ are given by \eqref{eq:AN-intro} and
\eqref{eq:TN-intro}.
\end{proposition}

\begin{proof}
Lemma~\ref{lem:BA-points} and Riemann's vanishing theorem show that all
apparent poles of the four theta quotients in
\eqref{eq:B11N-rigorous}--\eqref{eq:B22N-rigorous} are removable.  We next
check the jumps.  From \eqref{eq:Abel-jump-Sigmaj} and theta
quasi-periodicity, for $k\in\Sigma_j$, $j\ge1$,
\begin{equation}\label{eq:theta-jump-check}
 \frac{\Theta(-\mathbf Z_N+J_{N,+}(k)+\mathbf d_N)}
      {\Theta(J_{N,+}(k)+\mathbf d_N)}
 =\ee^{-2\pi\ii(\mathbf Z_N)_j}
 \frac{\Theta(-\mathbf Z_N-J_{N,-}(k)+\mathbf d_N)}
      {\Theta(-J_{N,-}(k)+\mathbf d_N)}.
\end{equation}
The same identity without the exponential factor holds on $\Sigma_0$ by
\eqref{eq:Abel-jump-Sigma0}.  Combining \eqref{eq:theta-jump-check} with
$\gamma_{N,+}=\ii\gamma_{N,-}$ gives
\begin{equation*}
 F_{N,+}=F_{N,-}
 \begin{pmatrix}
 0&-\ii\ee^{\ii\vartheta_{N,j}}\\
 -\ii\ee^{-\ii\vartheta_{N,j}}&0
 \end{pmatrix}
 \quad(k\in\Sigma_j),
\end{equation*}
and gives the jump $-\ii\sigma_1$ on $\Sigma_0$.  The factor
\eqref{eq:normalizerN-rigorous} yields $F_N(k)=I+O(k^{-1})$ at infinity.
The Schwarz symmetry follows from the real branch points and the real
canonical basis.  These properties verify RHP~\ref{rhp:kmodel}; uniqueness
follows by the usual determinant and Liouville argument.

For $|z|>1$, formula \eqref{eq:YN-from-FN} is the direct inverse of \eqref{eq:F-def}.  For $|z|<1$ it is obtained from the reciprocal symmetry
$Y_N(z)=z^{-1}Y_N(z^{-1})\sigma_1$.  In particular, as $z\to0$, $k(z)\to\infty$ and the second line gives $Y_N(z)=z^{-1}\sigma_1+O(1)$; a single global square-root formula would not have the correct zero singularity.  Thus \eqref{eq:YN-from-FN} reverses the Joukowski folding and restores the prescribed $z=0$ behavior.  Expanding the Baker--Akhiezer matrix at
$\infty_+$ and using \eqref{eq:dN-Abel-identity} gives
\eqref{eq:YN-expansion}--\eqref{eq:YN21}.  The coefficient of the algebraic
factor is the trace constant
$\frac12\sum_{m=0}^{N}(e_{2m}-e_{2m+1})=\AN$, while the remaining quotient is
exactly $\TN(\mathbf Z_N)$.
\end{proof}

\section{Long-time asymptotic theorem}\label{sec:theorem}

A compact subset of a sector is called admissible if it has positive distance from all critical velocities $v_m$.  All estimates below are uniform for $\xi=x/t$ in admissible compact subsets.

\begin{theorem}[Full arbitrary-genus dNLS soliton gas]\label{thm:main}
Let Assumptions~\ref{ass:data} and~\ref{ass:steepest} hold, let the endpoints satisfy \eqref{eq:endpoints-intro}, and let $q(x,t)$ be the uniquely defined solution reconstructed from RHP~\ref{rhp:full} by Proposition~\ref{prop:solvability}.  Then, as $t\to+\infty$ with $\xi=x/t$ fixed away from the critical velocities,
\begin{equation}\label{eq:main-theorem}
 q(x,t)=\AN\ee^{-2\ii t g_{N,\infty}(\xi)}
 \delta_{N,\infty}(\xi)^{-2}
 \TN\!\left(\frac{t\Omega_N(\xi)+\Delta_N(\xi)}{2\pi}\right)
 +\mathcal E_N(x,t).
\end{equation}
More precisely:
\begin{enumerate}[label=(\roman*)]
\item If $\xi\in\mathcal P_\ell$, $\ell=0,\ldots,N$, the Szeg\H{o} datum is
\begin{equation}\label{eq:ell-pure-theorem}
 \begin{aligned}
 \ell_\xi&=-\log c_r&&\text{on }\Gamma_j,\quad j\le\ell,\\
 \ell_\xi&= \log c_\rho&&\text{on }\Gamma_j,\quad j>\ell,
 \end{aligned}
\end{equation}
with the empty groups omitted, and
\begin{equation}\label{eq:error-pure-theorem}
 \mathcal E_N(x,t)=O(t^{-1}).
\end{equation}
\item If $\xi\in\mathcal M_\ell$, $\ell=1,\ldots,N$, let $\alpha_\ell(\xi)$ be determined by \eqref{eq:alpha-intro}.  The Szeg\H{o} datum is
\begin{equation}\label{eq:ell-mixed-theorem}
 \ell_\xi(z)=
 \begin{cases}
 -\log c_r(z),&z\in\Gamma_j,\ j<\ell,\\
 -\log c_r(z),&z\in(\eta_{2\ell-1},\alpha_\ell),\\
 \log c_\rho(z),&z\in(\alpha_\ell,\eta_{2\ell}),\\
 \log c_\rho(z),&z\in\Gamma_j,\ j>\ell,
 \end{cases}
\end{equation}
and
\begin{equation}\label{eq:error-mixed-theorem}
 \mathcal E_N(x,t)=O(t^{-1/2}).
\end{equation}
\end{enumerate}
In every sector the leading term is a fixed genus-$N$ theta-functional solution on \eqref{eq:curve-intro}.  There is no genus cascade and no planar leading sector.
\end{theorem}

\begin{remark}[Why every sector has the same genus]\label{rem:same-genus}
The full jump contains one growing exponential regardless of the sign of the bare phase.  A triangular factorization is therefore required on every spectral arc.  Proposition~\ref{prop:full-reduction} shows that after the triangular factors are moved to exponentially small lens lips, the same constant central jump remains on all $N$ upper arcs.  This is the precise reason that the reduced model has fixed genus $N$ in every one of the $2N+1$ sectors.
\end{remark}

\section{Proof of Theorem~\ref{thm:main}}\label{sec:proof}

\subsection{Global deformation}

The transformation $M\mapsto S$ is given by \eqref{eq:S-transform}.  Proposition~\ref{prop:signs} selects the decaying triangular factor on every component of every band.  The all-band transformation $S\mapsto T$ is given by \eqref{eq:T-transform}.  In a mixed sector the active band is opened crosswise at $\alpha_\ell$; all $j>\ell$ bands are also opened with the $\rho$-factorization.  Proposition~\ref{prop:full-reduction} then reduces the problem to the fixed model RHP~\ref{rhp:zmodel}, up to exponentially small jumps on all lens lips.

\subsection{Local parametrices}

Place discs around the $4N$ fixed endpoints $\eta_m^{\pm1}$, $m=1,\ldots,2N$.  Near a fixed endpoint, $\phi_N$ has square-root behavior.  A conformal coordinate of the form
\begin{equation}\label{eq:Bessel-coordinate}
 \zeta=-t^2\bigl(\phi_N(z;\xi)-\phi_N(\eta_m;\xi)\bigr)^2
\end{equation}
reduces the local problem to the modified-Bessel model.  The matching error is $I+O(t^{-1})$.  This is the standard hard-edge construction; see \cite{KuijlaarsEtAl2004,BertolaWangYanZhu2026}.

In a mixed sector $\mathcal M_\ell$, add discs around $\alpha_\ell(\xi)$ and $\alpha_\ell(\xi)^{-1}$.  Since
\begin{equation}\label{eq:switch-simple}
 \phi_N'(\alpha_\ell;\xi)=0,
 \qquad \phi_N''(\alpha_\ell;\xi)\neq0,
\end{equation}
there is a conformal coordinate
\begin{equation}\label{eq:PC-coordinate}
 \zeta=t^{1/2}(z-\alpha_\ell)\psi_\ell(z;\xi),
 \qquad \psi_\ell(\alpha_\ell;\xi)\neq0.
\end{equation}
Proposition~\ref{prop:Szego} separates the local scalar factor into an analytic nonzero factor and the power $(z-\alpha_\ell)^{\ii\nu_\ell}$, while the phase has the quadratic expansion
\[
 \phi_N(z;\xi)-\phi_N(\alpha_\ell;\xi)
 =\frac12\phi_N''(\alpha_\ell;\xi)(z-\alpha_\ell)^2
 +O((z-\alpha_\ell)^3).
\]
After analytic conjugations, the four crossed lens rays and the two portions of the original band give the standard six-ray parabolic-cylinder local RHP.  Its parameter is
\begin{equation}\label{eq:nu}
 \nu_\ell(\xi)=-\frac{1}{2\pi}
 \log\left(c_r(\alpha_\ell)c_\rho(\alpha_\ell)\right)
 =-\frac{1}{2\pi}\log\left(
 \frac{4r(\alpha_\ell)\rho(\alpha_\ell)}
 {(1+r(\alpha_\ell)\rho(\alpha_\ell))^2}\right)>0.
\end{equation}
The large-$\zeta$ expansion of the parabolic-cylinder model gives the matching error $I+O(t^{-1/2})$; see \cite{BoutetLenellsShepelsky2022,YanGengWei2026,FullCH2026}.

\subsection{Small-norm error problem}

Let $P^{(\infty)}=Y_N$ be the outer model and let $P$ denote the global parametrix obtained by inserting the endpoint and switching-point local parametrices.  Define
\begin{equation}\label{eq:error}
 E(z)=T(z)P(z)^{-1}.
\end{equation}
On every lens lip away from the discs,
\begin{equation}\label{eq:error-lips}
 J_E(z)=I+O(\ee^{-ct}).
\end{equation}
On endpoint-disc boundaries, $J_E=I+O(t^{-1})$.  In a mixed sector, the switching discs give $J_E=I+O(t^{-1/2})$.  Hence
\begin{align}
 E(z)&=I+\frac{E_1(x,t)}{z}+O(z^{-2}),\label{eq:E-expansion}\\
 E_1(x,t)&=O(t^{-1})&&\text{in }\mathcal P_\ell,\label{eq:E-pure}\\
 E_1(x,t)&=O(t^{-1/2})&&\text{in }\mathcal M_\ell.\label{eq:E-mixed}
\end{align}

\subsection{Reconstruction}

Undoing the transformations gives
\begin{equation}\label{eq:undo}
 M(z)=\delta_{N,\infty}^{\sigma_3}\ee^{\ii t g_{N,\infty}\sigma_3}
 E(z)P(z)\ee^{-\ii t g_N(z;\xi)\sigma_3}
 \delta_N(z;\xi)^{-\sigma_3}.
\end{equation}
Using \eqref{eq:gN-infty}, \eqref{eq:delta-sym}, and Proposition~\ref{prop:model}, the coefficient of $z^{-1}$ in the $(2,1)$ entry is
\begin{equation}\label{eq:leading-reconstruction}
 \AN\ee^{-2\ii t g_{N,\infty}(\xi)}\delta_{N,\infty}(\xi)^{-2}
 \TN\!\left(\frac{t\Omega_N(\xi)+\Delta_N(\xi)}{2\pi}\right),
\end{equation}
up to the contribution of $E_1$.  Equations \eqref{eq:E-pure}--\eqref{eq:E-mixed} and the reconstruction formula \eqref{eq:reconstruction} prove Theorem~\ref{thm:main}.

\section{Concluding remarks}

The full arbitrary-genus dNLS gas differs from the half gas at the global deformation level.  The active arc in a mixed sector is split and opened crosswise, but the later spectral arcs do not disappear: they are opened with the second triangular factorization and remain as constant-jump contours in the outer model.  This gives a fixed genus-$N$ finite-gap leading term in every sector.  The only sector dependence is carried by the Szeg\H{o} vector $\Delta_N$, the scalar $\delta_{N,\infty}$, and, in mixed sectors, the parabolic-cylinder local correction.

The analysis excludes neighborhoods of the endpoint velocities $v_m$.  At a critical ray the switching point collides with a fixed endpoint, and a separate transition parametrix is required.  

\section*{Data availability}
No new data were created or analyzed in this study.

\section*{Conflict of interest}
The authors declare no conflict of interest.

\section*{Acknowledgments}
This work is supported by the National Natural Science Foundation of China (Grant Nos. 12471234, 12401320, 12471240) and Science Foundation of Henan Academy of Sciences (Grant No. 20252319002).

\end{document}